\documentclass[journal,a4paper,romanappendices]{IEEEtran}
\usepackage{graphicx} % Required for inserting images
\usepackage{amsmath}
\usepackage{mathrsfs}
\usepackage{amssymb}
\usepackage{overpic}
\usepackage{algorithm}
\usepackage{algorithmic}
\usepackage{array}
\usepackage{graphicx}
\usepackage{color}
\usepackage{epstopdf}
\usepackage{acronym}
\usepackage{enumerate}
\definecolor{BLUE}{rgb}{0,0,1}

\acrodef{aoa}[AOA]{angle-of-arrival}
\acrodef{acf}[ACF]{autocorrelation function}
\acrodef{bcrb}[BCRB]{Bayesian Cram\'{e}r-Rao bound}
\acrodef{bp}[BP]{belief propagation}
\acrodef{cdi}[CDI]{cooperative dilution intensity}
\acrodef{cl}[CL]{cooperative localization}
\acrodef{cdf}[CDF]{cumulative distribution function}
\acrodef{crb}[CRB]{Cram\'{e}r-Rao bound}
\acrodef{crlb}[CRLB]{Cram\'{e}r-Rao lower bound}
\acrodef{dof}[DoF]{degree of freedom}
\acrodef{dct}[DCT]{discrete cosine transform}
\acrodef{dpeb}[DPEB]{directional position error bound}
\acrodef{fim}[FIM]{Fisher information matrix}
\acrodef{efim}[EFIM]{equivalent Fisher information matrix}
\acrodef{icrb}[ICRB]{inverse CRB}
\acrodef{iid}[i.i.d.]{independently and identically distributed}
\acrodef{isac}[ISAC]{integrated sensing and communication}
\acrodef{mse}[MSE]{mean-squared error}
\acrodef{pdf}[PDF]{probability density function}
\acrodef{peb}[PEB]{position error bound}
\acrodef{speb}[SPEB]{squared position error bound}
\acrodef{pll}[PLL]{phase-locked loop}
\acrodef{rbs}[RBS]{reference broadcast synchronization}
\acrodef{rhs}[RHS]{right hand side}
\acrodef{rii}[RII]{ranging information intensity}
\acrodef{rss}[RSS]{received signal strength}
\acrodef{rc}[RC]{ranging coefficient}
\acrodef{toa}[TOA]{time-of-arrival}
\acrodef{tdoa}[TDOA]{time-difference-of-arrival}
\acrodef{tpsn}[TPSN]{time synchronization protocol for sensor network}
\acrodef{vmp}[VMP]{variational message passing}
\acrodef{wsn}[WSN]{wireless sensor network}
\acrodef{dio}[DIO]{distance-information-only}
\acrodef{aio}[AIO]{angle-information-only}
\acrodef{saaf}[SAAF]{squared array aperture function}
\acrodef{snc}[S\&C]{sensing and communications}
\acrodef{uoa}[UOA]{uniformly oriented array}
\acrodef{rgg}[RGG]{random geometric graph}
\acrodef{rms}[RMS]{root-mean-square}
\acrodef{snr}[SNR]{signal-to-noise ratio}
\acrodef{eoc}[EoC]{efficiency of cooperation}
\acrodef{npi}[NPI]{nominal position information}
\acrodef{gnss}[GNSS]{global navigation satellite system}
\acrodef{mimo}[MIMO]{multiple-input multiple-output}
\acrodef{mcs}[MCS]{minimally constrained system}
\acrodef{zzb}[ZZB]{Ziv-Zakai bound}
\acrodef{wwb}[WWB]{Weiss-Weinstein lower bound}
\acrodef{nlos}[NLOS]{non-light-of-sight}
\acrodef{mmse}[MMSE]{minimum mean squared error}
\acrodef{uav}[UAV]{unmanned aerial vehicle}
\acrodef{ppp}[PPP]{Poisson point process}
\acrodef{bpp}[BPP]{binomial point process}
\acrodef{cln}[CLN]{cooperative location-aware network}
\acrodef{pdr}[PDR]{pedestrian dead reckoning}
\acrodef{ml}[ML]{maximum likelihood}
\acrodef{map}[MAP]{maximum \textit{a posteriori}}
\acrodef{ofdm}[OFDM]{orthogonal frequency division multiplexing}
\acrodef{otfs}[OTFS]{orthogonal time frequency space}
\acrodef{sc}[SC]{single carrier}
\acrodef{cp}[CP]{cyclic prefix}
\acrodef{pacs}[P-ACS]{periodic autocorrelation sequence}
\acrodef{eisl}[EISL]{expected integrated sidelobe level}
\acrodef{dft}[DFT]{discrete Fourier transform}
\acrodef{6g}[6G]{sixth-generation}
\acrodef{drt}[DRT]{deterministic-random tradeoff}
\acrodef{psk}[PSK]{phase shift keying}
\acrodef{qpsk}[QPSK]{quadrature phase shift keying}
\acrodef{bpsk}[BPSK]{binary phase shift keying}
\acrodef{ask}[ASK]{amplitude shift keying}
\acrodef{apsk}[APSK]{amplitude phase shift keying}
\acrodef{qam}[QAM]{quadrature amplitude modulation}
\acrodef{pam}[PAM]{pulse amplitude modulation}
\acrodef{ba}[BA]{Blahut-Arimoto}
\acrodef{islr}[ISLR]{integrated sidelobe level ratio}
\acrodef{mpdm}[MPDM]{multiset-partition distribution matching}
\acrodef{ess}[ESS]{enumerative sphere shaping}
\acrodef{cdma}[CDMA]{code division multiple access}
\acrodef{afdm}[AFDM]{affine frequency division multiplexing}
\acrodef{mi}[MI]{mutual information}
\acrodef{aacf}[A-ACF]{aperiodic autocorrelation function}
\acrodef{fft}[FFT]{fast Fourier transform}
\acrodef{mb}[MB]{Maxwell-Boltzmann}
\acrodef{fec}[FEC]{forward error-correction code}
\acrodef{gpas}[G-PAS]{generalized PAS}
\acrodef{kld}[KLD]{Kullback-Leibler divergence}
\acrodef{jhmess}[JHM-ESS]{joint hard-moment ESS}
\acrodef{sess}[S-ESS]{scalarized ESS}
\acrodef{bler}[BLER]{block error rate}
\acrodef{fbl}[FBL]{finite blocklength}

\acrodef{acfk}[ACFK]{auto-correlation function keying}
\acrodef{afk}[AFK]{ambiguity function keying}
\acrodef{af}[AF]{ambiguity function}
\acrodef{acf}[ACF]{auto-correlation function}
\acrodef{pacf}[P-ACF]{periodic auto-correlation function}
\acrodef{ber}[BER]{bit error rate}
\acrodef{ser}[SER]{symbol error rate}
\acrodef{dpaf}[DP-AF]{discrete periodic ambiguity function}
\acrodef{fstaf}[FST-AF]{fast-slow-time ambiguity function}
\acrodef{fst}[FST]{fast-slow-time}
\acrodef{ici}[ICI]{inter-carrier interference}
\acrodef{esl}[ESL]{expected sidelobe level}
\acrodef{psl}[PSL]{peak sidelobe level}
\acrodef{pslr}[PSLR]{peak sidelobe level ratio}
\acrodef{isi}[ISI]{inter-symbol interference}
\acrodef{zf}[ZF]{zero forcing}
\acrodef{llr}[LLR]{log-likelihood ratio}
\acrodef{awgn}[AWGN]{additive white Gaussian noise}
\acrodef{sir}[SIR]{signal-to-interference ratio}
\acrodef{cfar}[CFAR]{constant false-alarm rate}
\acrodef{cacfar}[CA-CFAR]{cell-averaging constant false-alarm rate}
\acrodef{tx}[Tx]{transmitter}
\acrodef{rx}[Rx]{receiver}
\acrodef{lti}[LTI]{linear time-invariant}
\acrodef{mf}[MF]{matched-filtering}
\acrodef{mgf}[MGF]{moment generating function}
\acrodef{pcs}[PCS]{probabilistic constellation shaping}
\acrodef{pas}[PAS]{probabilistic amplitude shaping}
\acrodef{ldpc}[LDPC]{low-density parity-check}
\acrodef{ccdm}[CCDM]{constant composition distribution matching}
\acrodef{roc}[ROC]{receiver operating characteristic}
\acrodef{cdf}[CDF]{cumulative distribution function}
\acrodef{csirs}[CSI-RS]{channel state information reference signals}
\acrodef{dmrs}[DMRS]{demodulation reference signals}
\acrodef{zcz}[ZCZ]{zero-correlation zone}
\acrodef{lcz}[LCZ]{low-correlation zone}
\acrodef{2dfft}[2D-FFT]{two-dimensional fast Fourier transform}

\acrodefplural{dof}[DoFs]{degrees of freedom}

\usepackage{winsnotation}
\usepackage{algorithm}
\usepackage{algorithmic}
\usepackage{makecell}
\usepackage{subfig}
\usepackage{cite}
\usepackage{flushend}

\newtheorem{assumption}{Assumption}
\newtheorem{problem}{Problem}
\newtheorem{definition}{Definition}
\newtheorem{proposition}{Proposition}
\newtheorem{lemma}{Lemma}
\newtheorem{theorem}{Theorem}
\newtheorem{remark}{Remark}

\usepackage{tikzsymbols}

\title{How Does Probabilistic Constellation Shaping Control Ranging Sidelobe Levels in ISAC?}
\author{Zexuan Jing, \IEEEmembership{Member, IEEE}, Yifeng Xiong, \IEEEmembership{Member, IEEE}, Fan Liu, \IEEEmembership{Senior Member, IEEE}, Xuehua Li, \IEEEmembership{Member, IEEE}, and Jianhua Zhang, \IEEEmembership{Fellow, IEEE}
\thanks{This work was supported in part by the National Natural Science Foundation of China (NSFC) under Grant U25B2009, Grant 62522107, and Grant 62331023. (\textit{Corresponding author: Yifeng Xiong.})}
}

\begin{document}
\maketitle
\begin{abstract}
\Ac{pcs}, originally designed to fill the gap between finite uniform constellations and continuous capacity-achieving input distributions, has recently emerged as a central tool in adjusting the \ac{eisl} in \ac{isac}. In this paper, we show that \ac{pcs} goes beyond \ac{eisl}: A large class of \ac{pcs} methods, termed as type-based \ac{pcs}, are capable of providing precise control over codeword-wise \ac{islr} for systems employing \ac{ofdm}, thus yielding a stronger sensing performance guarantee than \ac{eisl} control. We further show that \ac{ofdm} is almost the only unitary modulation basis whose codeword-wise \ac{islr} can be precisely controlled by type-based \ac{pcs}. Based upon the analytical results, practical type-based \ac{pcs} methods compatible with the \ac{pas} framework are conceived, which achieve higher communication rate than existing \ac{isac}-oriented \ac{pcs} methods under identical \ac{islr} constraints.
\end{abstract}

\section{Introduction}
\Acf{isac} is emerging as an important capability of future wireless networks, enabling the same radio infrastructure to support both information delivery and environmental awareness. By turning communication transmissions into probes of the surrounding environment, \ac{isac} is expected to facilitate applications such as autonomous transportation, industrial automation, smart environments, and low-altitude services \cite{Chafii2023CST, saad2019vision, 9606831}. 

In a practical wireless frame, freely designed reference signals, such as preambles and pilots, usually occupy only a limited fraction of the available time-frequency resources \cite{pilot_and_data}, while most resources are devoted to random data payload symbols. In monostatic or cooperative \ac{isac}, the transmitted data sequence is directly available to the sensing processor and may therefore be used as a reference for matched filtering \cite{iceberg}. Making reliable use of these payload signals is thus a central problem in communication-centric \ac{isac}. However, the data symbols have to convey information and are inherently random. Consequently, their sensing behavior varies from one transmission to another, a phenomenon known as the \ac{drt} \cite{tradeoff_tit}. 

From the perspective of sensing, the quality of a probing waveform is commonly evaluated through its autocorrelation function or ambiguity function \cite{radarbook}. To elaborate, a target produces a delayed copy of the transmitted waveform, and the corresponding matched filter output contains a mainlobe surrounded by sidelobes. A large sidelobe generated by a strong target may mask the mainlobe of a nearby weak target. One important measure of this effect is the \acf{islr}, defined as the ratio between the aggregate sidelobe energy and the mainlobe energy. It thus reflects how well targets of different strengths can be separated in the range domain. When random communication symbols are used for sensing, however, the autocorrelation function and the resulting \ac{islr} also become random quantities.

To obtain a tractable ensemble-level description, existing works proposed the \ac{eisl}, namely the expectation of the aggregate sidelobe energy, as a ranging metric for random communication waveforms \cite{ofdm_optimal,iceberg}. This metric makes it possible to compare signaling schemes without accounting for each data realization. In this context, it is shown that, under the \ac{pacf} model, \ac{ofdm} minimizes the \ac{eisl} among unitary modulation bases \cite{ofdm_optimal}. 

\Acf{pcs} has been applied in this context for controlling this average sensing behavior \cite{pcs_isac,pcs_filter,isac_pcs_pas,pcs_prototype}. The technique was originally developed to close the gap between finite uniform constellations and the continuous capacity-achieving input distributions \cite{pas}. By adjusting the probabilities of constellation symbols, \ac{pcs} reshapes the energy as well as higher-order statistics of the input distribution. In particular, the fourth-order moment acts as a key parameter for controlling the \ac{eisl}. In light of this, recent works have employed \ac{pcs} to approach the communication-optimal input distribution subject to an \ac{eisl} constraint \cite{pcs_isac,isac_pcs_pas}. 

Nevertheless, existing studies treat \ac{pcs} primarily as a statistical tool for adjusting the \ac{eisl}, without characterizing the intrinsic limit of the sidelobe control that it can provide. In particular, it is unclear whether \ac{pcs} can offer a guarantee stronger than the ensemble mean (i.e., the \ac{eisl}), such as an upper bound that holds for every finite-length codeword and remains useful after coherent integration. Moreover, conventional \ac{pcs} and \ac{pas} are designed around communication-optimal input distributions rather than the joint moment constraints arising in \ac{isac}. Their usual factorized or rotational-symmetric constructions may therefore incur unnecessary finite-length rate loss under a given sensing constraint \cite{isac_pcs_pas}, leading to an unfavorable sidelobe-rate tradeoff. A satisfactory solution to this issue, should both reveal the strongest sidelobe guarantee inherent to \ac{pcs}, and translate that guarantee into practical shaping algorithms.

Against this background, in this paper, we aim at answering two key questions: 
\begin{enumerate}
\item What codeword-level sidelobe guarantees can \ac{pcs} provide, beyond the ensemble-level \ac{eisl}, and where is the boundary of this capability? 
\item How can practical \ac{pcs} schemes realize stronger sidelobe control while retaining favorable finite-length communication performance? 
\end{enumerate}

Specifically, our main contributions are summarized as follows.
\begin{itemize}
\item We formulate a broad class of existing \ac{pcs} methods as \emph{type-based \ac{pcs}}. This framework encompasses the conceptual forms of \ac{ccdm}, \ac{mpdm}, shell mapping, and \ac{ess}. We show that type-based \ac{pcs} provides codeword-wise \ac{islr} control if and only if the \ac{islr} is invariant to every symbol permutation, which converts the search for a maximum-\ac{islr} guarantee into an enquiry about the structure of the modulation basis.
\item For the \ac{pacf} and $N\geq 7$, we show that all unitary modulation bases that permit type-based codeword-wise \ac{islr} control belong to an explicitly characterized restricted family, with \ac{ofdm} being a representative example. On other other hand, for the \ac{aacf} and $N\geq 4$, no unitary modulation basis permits such control. In this case, type-based \ac{pcs} can thus only control ensemble statistics such as the \ac{eisl}.
\item We propose a method, termed as \emph{sifted shifting}, to extend the individual-codeword guarantee to coherent integration. By selecting a frequency-domain circular shift for each symbol of an $M$-symbol \ac{ofdm} block, the method preserves its type and individual-symbol \ac{pacf}-based \ac{islr}, while guaranteeing an integrated \ac{islr} no larger than $\gamma/M_{\rm eff}$ whenever every \ac{ofdm} symbol satisfies an individual \ac{islr} upper bound $\gamma$, where $M_{\rm eff}$ is the effective integration length and equals $M$ for equal-energy symbols.
\item We introduce the framework quadrant \ac{pas}, which enables joint shaping over square \ac{qam} while retaining uniform sign bits and compatibility with binary \ac{fec}. Within this framework, we adapt \ac{ccdm} and \ac{ess} to \ac{isac}-oriented input distributions. Within the scope of \ac{ess}, we develop the \ac{jhmess} and \ac{sess} constructions, which trade deterministic sidelobe control against rate and memory cost. Numerical simulations suggest that the proposed \ac{pcs} methods yield reduced finite-length rate loss compared to existing methods, under identical \ac{islr} constraints.
\end{itemize}

The remainder of this paper is organized as follows. Sec.~\ref{sec:formulation} formulates the communication and ranging metrics and defines type-based \ac{pcs}. Sec.~\ref{sec:islr} characterizes codeword-wise \ac{islr} controllability for the \ac{pacf} and \ac{aacf}. Sec.~\ref{sec:ss} develops sifted shifting for coherent integration. Sec.~\ref{sec:design} presents the practical \ac{pcs} constructions and their cost-performance tradeoffs. Sec.~\ref{sec:numerical} reports the numerical results, and Sec.~\ref{sec:conclusions} concludes the paper.

\subsubsection*{Notations}
Throughout this paper, $a$, $\V{a}$, $\M{A}$, and $\Set{A}$ represent scalars, vectors, matrices and sets, respectively. The $m$-dimensional vector of zeros (resp. ones) is denoted by $\V{0}_{m}$ (resp. $\V{1}_{m}$). The $m$-by-$m$ identity matrix is denoted by $\M{I}_{m}$. These subscripts are omitted if they are clear from the context. The notation $[\M{A}]_{i,j}$ denotes the $(i,j)$-th entry of matrix $\M{A}$. The notations $|\V{x}|$ and $\sqrt{\V{x}}$ denote the vector containing the entrywise magnitude and the entrywise squared magnitudes of $\V{x}$, respectively. ${\rm{diag}}(\V{x})$ denotes the diagonal matrix constructed by placing $\V{x}$ on its main diagonal. $|\Set{A}|$ denotes the cardinality of the set $\Set{A}$. The notation $\|\V{x}\|_p$ represents the $\ell_p$ norm of the vector $\V{x}$. The direct sum of linear subspaces is denoted by $\Set{A}\oplus\Set{B}$, and ${\rm span}(\cdot)$ denotes the linear space spanned by its arguments.

\section{Problem Formulation}\label{sec:formulation}
\subsection{Model and performance metrics}
In this paper we consider the simplistic \ac{awgn} communication channel model\footnote{For more realistic frequency-selective channels, joint second-order and fourth-order moment allocation methods can be applied to achieve a more favorable sensing-communication performance tradeoff \cite{ccom_kurtosis} Here, we focus on the \ac{awgn} model for clarity of exposition, while the proposed analytical results and methods can also be extended to frequency-selective channels.}
\begin{equation}
\V{y}=\sqrt{\rho}\V{x}+\V{z},
\end{equation}
where $\rho$ denotes the \ac{snr}, and $\V{z}\sim\mathcal{CN}(0,\M{I}_N)$ denotes the complex \ac{awgn}. The transmitted vector $\V{x}$ is obtained by modulating $N$ \ac{iid} symbols via a unitary modulation basis $\M{U}$, namely 
\begin{equation}\label{unitary_modulation}
\V{x}=\M{U}\V{s}=\sum_{n=1}^N s_n\V{u}_n.
\end{equation}
By varying $\M{U}$, the model \eqref{unitary_modulation} can be used to represent a large amount of existing modulation schemes, including \ac{ofdm}, \ac{sc}, \ac{cdma}, \ac{otfs}, and \ac{afdm} \cite{ofdm_optimal}.

The symbols $\{s_n\}$ are drawn from a finite constellation set $\Set{S}$, e.g., \ac{qam} or \ac{psk}, etc. From the communication perspective, a natural design objective is the \ac{mi}, thus yielding the generic \ac{isac}-oriented input distribution optimization problem
\begin{subequations}
\begin{align}
\max_{p_{\V{s}}(\V{s})}&~~I(\V{s};\V{y})\\
{\rm s.t.}&~~p_{\V{s}}(\V{s})\in\Set{P}(\Set{S})\cap\Set{S}_{\rm c}\cap \Set{S}_{\rm s},
\end{align}
\end{subequations}
where $I(\V{s};\V{y})$ represents the \ac{mi} between $\V{s}$ and $\V{y}$, $\Set{P}(\Set{S})$ denotes the set of probabilistic distributions over $\Set{S}$, $\Set{S}_{\rm c}$ denotes the feasible region of sensing-oriented constraints, and $\Set{S}_{\rm c}$ denotes the feasible region of other design constraints (e.g., power constraints). Once an optimal input distribution $p_{\V{s}}(\V{s})$ is solved, one may obtained a finite blocklength approximation of it using \ac{pcs} methods. 

In the context of ranging, sensing performance metrics are generated from \ac{aacf} and \ac{pacf} of the transmitted signal $\V{x}$. In particular, denote the $k$-th shift and periodic shift matrices as
\begin{equation}
\M{J}_k=\begin{bmatrix}
\M{0} & \M{I}_{N-k}\\
\M{0} & \M{0}
\end{bmatrix},\qquad\widetilde{\M{J}}_k=\begin{bmatrix}
\M{0} & \M{I}_{N-k}\\
\M{I}_k & \M{0}
\end{bmatrix},
\end{equation}
respectively, the \ac{aacf} and \ac{pacf} can be expressed as
\begin{equation}
r_k=\V{x}^{\rm H}\M{J}_k\V{x},\qquad\widetilde{r}_k=\V{x}^{\rm H}\widetilde{\M{J}}_k\V{x},
\end{equation}
respectively, for $k=0,\dotsc,N-1$. The widely applied \ac{islr} \cite{isl_good} depicting the inter-target interference can be defined based on either \ac{aacf} or \ac{pacf} as follows
\begin{equation}
{\rm ISLR}=\frac{\sum_{k=1}^{N-1}|a_k|^2}{|a_0|^2},
\end{equation}
where $a_k$ is either $r_k$ or $\widetilde{r}_k$. Note that \ac{islr} depends on a single realization of $\V{x}$ and hence on a single codeword $\V{s}$. When precise control over realizations is infeasible, one may resort to the expectation over the entire codebook. In particular, under a constant power constraint $\mathbb{E}\{|s_n|^2\}=1,~\forall n$, it is reasonable to consider the \ac{eisl} \cite{ofdm_optimal}
\begin{equation}
{\rm EISL} = \sum_{k=1}^{N-1}\mathbb{E}\{|a_k|^2\}.
\end{equation}
As we shall see later, codebook-level control (i.e., over \ac{eisl}) and codeword-level control (i.e., over \ac{islr}) can lead to rather different requirements on the underlying modulation basis and \ac{pcs} methods.

\subsection{Type-based \ac{pcs}}
The type of a length-$N$ codeword $\V{s}$, also referred to as an $N$-type, is defined as
\begin{equation}
\hat{P}_{\V{s}}(a)=\frac{1}{N}\left|\{k:s_k=a\}\right|,~a\in\Set{S}.
\end{equation}
We denote the set of all $N$-types over the constellation set $\Set{S}$ by $\Set{P}_N(\Set{S})$. The corresponding type class given by
\begin{equation}
\Set{T}_N(q)=\{\V{s}\in\Set{S}^N:\hat{P}_{\V{s}}=q\},
\end{equation}
is the set containing all codewords of the same type $q$. A type-based \ac{pcs} method is characterized by a distribution
\begin{equation}\label{type_based_def}
P_N(\V{s})=\frac{\omega_N(q)}{|\Set{T}_N(q)|},~\V{s}\in\Set{T}_N(q),
\end{equation}
with $\omega_N(q)\geq 0$, $\sum_q \omega_N(q)=1$, from which a codeword is drawn by a two-step procedure:
\begin{enumerate}
\item Specify a type $q$ with probability $\omega_N(q)$;
\item Given $q$, draw uniformly from the type class $\Set{T}_N(q)$.
\end{enumerate}

The distribution \eqref{type_based_def} can be determined either directly, or indirectly via constraints imposed on a (potentially vector-valued) type function $g_N:~\Set{P}_N(\Set{S})\rightarrow \mathbb{R}^r$. In particular, denote the feasible region of the value of $g_N$ by $\Set{B}_N\subseteq \mathbb{R}^r$, the corresponding set of allowable codewords is given by
\begin{equation}
\Set{R}_N=\{\V{s}|g_N(\hat{P}_{\V{s}})\in\Set{B}_N\}.
\end{equation}
Observe that
\begin{equation}
\Set{R}_N = \bigcup_{q\in\Set{P}_N(\Set{S})\atop g_N(q)\in\Set{B}_N} \Set{T}_N(q).
\end{equation}
Various existing \ac{pcs} methods, in their conceptual forms, can be expressed as specific examples of type-based \ac{pcs}, as exemplified below.
\begin{enumerate}
\item \emph{\Ac{ccdm}} \cite{ccdm}: The \ac{ccdm} uses a single type. In other words, given a target type $q^\star$, the \ac{ccdm} is characterized by $\omega_N(q^\star)=1$ and $\Set{R}_N=\Set{T}_N(q^\star)$.
\item \emph{\Ac{mpdm}} \cite{mpdm}: The \ac{mpdm} chooses $L$ types $q^{(1)},\dotsc,q^{(L)}$, and select $m_{\ell}$ codewords from the $\ell$-th type as follows
\begin{equation}
\Set{C}_\ell \subseteq \Set{T}_N(q^{(\ell)}),~|\Set{C}_\ell|=m_\ell.
\end{equation}
Given a target type $q^\star$, the weights are thus determined by
\begin{equation}
P_N(\V{s})=\frac{\omega_N(q^{(\ell)})}{|\Set{T}_N(q^{(\ell)})|},~\sum_{\ell=1}^L\omega_N(q^{(\ell)})q^{(\ell)} = q^\star.
\end{equation}
\item \emph{Shell mapping \cite{sm}}: Shell mapping relies on a symbol-wise separable cost function
\begin{equation}
w(\V{s})=\sum_{i=1}^N w(s_i),
\end{equation}
which is indeed a type function since $w(\V{s})=N\langle \hat{P}_{\V{s}},w\rangle$. Given a maximum cost $t^\star$, the codebook of shell mapping is given by
\begin{subequations}
\begin{align}
\Set{R}_N&=\bigcup_{t\leq t^\star}\Set{H}_t, \\
\Set{H}_t&=\{\V{s}:w(\V{s})=t\}=\bigcup_{q:N\langle q,w\rangle=t}\Set{T}_N(q).
\end{align}
\end{subequations}
\item \emph{\Ac{ess}} \cite{ess,ess_short}: The \ac{ess} uses all codewords satisfying certain symbol-wise separable cost constraints in the sense that
\begin{equation}
\Set{R}_N=\bigcup_{q\in\Set{P}_N(\Set{S})\atop \V{c}(q)\preceq \V{c}^\star}\Set{T}_N(q),
\end{equation}
where $\V{c}^\star$ denotes the constraint vectors, and $\V{c}(q)$ denotes the cost vector, whose entries are linear functions of the type $q$.
\end{enumerate}

In practice, \ac{pcs} methods are typically fixed rate, in the sense that the size of the codebook is $2^K$ with $K$ being a positive integer. This would lead to an inevitable trimming of the type classes. Ideally, one may filter the ``best'' $2^K$ codewords from the type classes for a certain purpose, for example, to optimize the sensing performance. However, this is typically impractical since the size of type classes grows exponentially with $N$ \cite{it_cover}. Therefore, we make explicitly the following assumption.
\begin{assumption}[Uniform trimming]\label{asu:uniform}
All type-based \ac{pcs} methods, if have to select $2^K$ codewords from their raw type classes, can only select randomly according to uniform distributions over each type class.
\end{assumption}

\subsection{Known results and problems}
We are now ready to clarify what is known prior to this work and what is to be investigated. In \cite{ofdm_optimal}, it is shown that under the constant-power and rotational-symmetric assumptions on the input distribution, namely
\begin{equation}
\mathbb{E}\{s\}=0,~\mathbb{E}\{|s|^2\}=1,~\mathbb{E}\{s^2\}=0
\end{equation}
that apply to all $s_n$'s, the \ac{eisl} of \ac{aacf} and \ac{pacf} can be expressed as
\begin{subequations}
\begin{align}
\sum_{k=1}^{N-1}\!\mathbb{E}\{|r_k|^2\}\!&=\!\frac{N(N\!-\!\!1)}{2}\!\!+\!\!(\kappa\!-\!2)N\Big(\bigl\|\widetilde{\M{F}}_{2N}\M{U}\bigr\|_4^4\!-\!\frac{1}{2}\Big),\\
\sum_{k=1}^{N-1}\!\mathbb{E}\{|\widetilde{r}_k|^2\}\!&=\!N(N\!-\!1)\!+\!(\kappa\!-\!2)N(\|\M{F}_N\M{U}\|_4^4-1),
\end{align}
\end{subequations}
respectively, where $\M{F}_N$ denotes the $N$-point \ac{dft} matrix, $\widetilde{\M{F}}_{2N}$ contains the first $N$ columns of the $2N$-point \ac{dft} matrix $\M{F}_{2N}$, and $\kappa=\mathbb{E}\{|s|^4\}$ denotes the fourth-order moment of the input distribution. 

Observe that the quantity $\kappa$ serves as a key parameter in adjusting the \ac{eisl}, for a generic modulation basis $\M{U}$. It is further shown that \ac{ofdm}, corresponding to $\M{U}=\M{F}_N^{\rm H}$, yields the optimal \ac{eisl} for \ac{pacf} among all modulation bases \cite{ofdm_optimal}. Aligned with this framework, \ac{pcs} methods have been proposed to adjust the \ac{eisl}-rate tradeoff in \ac{isac} \cite{pcs_isac,pcs_filter,isac_pcs_pas}.

In this paper, we go one step further to investigate the codeword-level control capability of \ac{pcs} methods. In particular, we consider the following problems:
\begin{problem}[\ac{islr} controllability]
For type-based \ac{pcs} methods satisfying Assumption \ref{asu:uniform}, under what conditions, can they provide precise control over codeword-wise \ac{islr}?
\end{problem}
\begin{problem}[Practical \ac{pcs} design]
How can we design practical \ac{pcs} methods yielding favorable rate-\ac{islr} tradeoffs?
\end{problem}

\section{\ac{islr} Controllability}\label{sec:islr}
Let us first give a formal definition of controllability for codeword-wise performance metrics under type-based \ac{pcs} methods.
\begin{definition}[Type-based controllability]
Let $f_{\M{U}}(\V{s})$ be a function of codewords $\V{s}$, induced by the modulation basis $\M{U}$. We say that $f_{\M{U}}(\V{s})$ is type-based controllable, if there exists a function $h_{\M{U}}:\Set{P}_N(\Set{S})\rightarrow \mathbb{R}$ of types, such that
\begin{equation}\label{controllability}
f_{\M{U}}(\V{s})=h_{\M{U}}(\hat{P}_{\V{s}}),\qquad\V{s}\neq \V{0}.
\end{equation}
\end{definition}

Our central theoretical problem can thus be rephrased as investigating the type-based controllability of \ac{islr}. Next, we provide a convenient equivalent condition for the verification of type-based controllability.
\begin{proposition}[Type-based controllability is permutation invariance]\label{prop:perm}
For $\V{s}\in\mathbb{C}^N\backslash\{\V{0}\}$, the following are equivalent:
\begin{enumerate}
\item There exists $h_{\M{U}}$ such that \eqref{controllability} holds;
\item For any permutation matrix $\M{\Pi}$ and non-zero $\V{s}$, it holds that
\begin{equation}
f_{\M{U}}(\M{\Pi}\V{s})=f_{\M{U}}(\V{s}).
\end{equation}
\end{enumerate}
\begin{IEEEproof}
Please refer to Appendix \ref{ssec:proof_perm}.
\end{IEEEproof}
\end{proposition}

The remaining problem is to find conditions for $\M{U}$, under which the \ac{islr} is indeed permutation invariant. 

\subsection{The \ac{pacf} case}
In the \ac{pacf} case, by direct application of the Wiener-Khinchine theorem and the Parseval identity, we obtain immediately
\begin{equation}
{\rm ISLR} = \frac{N\|\M{V}\V{s}\|_4^4}{\|\V{s}\|_2^4}-1,
\end{equation}
where $\M{V}=\M{F}_N\M{U}$. Since the denominator is clearly permutation invariant, our problem can be recast as finding unitary $\M{V}$ satisfying
\begin{equation}\label{perm_sym_v}
\|\M{V}\M{\Pi}\V{s}\|_4 = \|\M{V}\V{s}\|_4,~~ \V{s}\in\mathbb{C}^N,~\forall \M{\Pi}\in\mathscr{S}_N,
\end{equation}
where $\mathscr{S}_N$ denotes the symmetric group over $N$ elements. To proceed, we need the following definition and lemma.
\begin{definition}[Monomial unitary matrices]
A unitary matrix $\M{U}\in\mathbb{U}(N)$ is referred to as a monomial unitary matrix, if it has exactly one non-zero entry in each row and each column. The group constituted by such matrices is denoted by $\mathsf{Mon}(N)$.
\end{definition}
\begin{lemma}[$\ell_4$-isometric unitary is monomial]\label{lem:monomial}
If $\M{U}\in\mathbb{U}(N)$ and $\|\M{U}\V{y}\|_4=\|\V{y}\|_4,~\forall\V{y}\in\mathbb{C}^N$, we have $\M{U}\in\mathsf{Mon}(N)$.
\begin{IEEEproof}
Please refer to Appendix \ref{ssec:proof_monomial}.
\end{IEEEproof}
\end{lemma}

Now, for each permutation matrix $\M{\Pi}$, we define
\begin{equation}
\M{V}_{\M{\Pi}} = \M{V}\M{\Pi}\M{V}^{\rm H}.
\end{equation}
Substituting $\V{s}^\prime=\M{V}\V{s}$ into \eqref{perm_sym_v}, we see that
\begin{equation}
\|\M{V}\M{\Pi}\M{V}^{\rm    H}\V{s}^\prime\|_4=\|\M{V}_{\M{\Pi}}\V{s}^\prime\|_4=\|\V{s}^\prime\|_4,
\end{equation}
thereby $\M{V}_{\M{\Pi}}$ is $\ell_4$-isometric, which further implies that
\begin{equation}
\M{V}\mathscr{S}_N\M{V}^{\rm H}\subseteq \mathsf{Mon}(N),
\end{equation}
according to Lemma \ref{lem:monomial}. This is equivalent to
\begin{equation}
\M{V}_{\M{\Pi}}\V{e}_k = \lambda_{\M{\Pi},k} \V{e}_{\sigma_{\M{\Pi}}(k)},
\end{equation}
where $\sigma_{\M{\Pi}}(k)$ denotes the output index corresponding to $k$ after the rearrangement induced by $\M{\Pi}$, and $\lambda_{\M{\Pi},k}$ satisfies $|\lambda_{\M{\Pi},k}|=1$. Next, let us denote the one-dimensional subspaces (lines) spanned by the rows of $\M{V}$ by $\Set{L}_k$, we have $\Set{L}_k = \M{V}^{\rm H}{\rm span}(\V{e}_k)$, from which we further obtain
\begin{align}
\M{\Pi}\Set{L}_k &= \M{\Pi} \M{V}^{\rm H}{\rm span}(\V{e}_k)\nonumber \\
&=\M{V}^{\rm H}\M{V}_{\M{\Pi}}{\rm span}(\V{e}_k) \nonumber \\
&=\M{V}^{\rm H}{\rm span}(\V{e}_{\sigma_{\M{\Pi}}(k)})=\Set{L}_{\sigma_{\M{\Pi}}(k)}.
\end{align}
This suggests that $\M{V}_{\M{\Pi}}$ is also a rearrangement of the lines $\{\Set{L}_k\}$. Once the lines are characterized, the only remaining degrees of freedom in $\M{V}$ are rearrangements of the rows, and a unimodular scalar multiplied on each of the rows, which is exactly a monomial unitary matrix multiplied from the left. 

Our next task is to characterize the lines. Note that any orbital of the symmetric group action over the lines would span a $\mathscr{S}_N$-invariant subspace. Furthermore, these subspaces are mutually orthogonal, and their direct sum constitutes the entire $\mathbb{C}^N$. The classical representation theory of $\mathscr{S}_N$ \cite{sagan2001symmetric} suggests the following invariant subspace decomposition
\begin{equation}
\mathbb{C}^N = \Set{E}\oplus \Set{E}^\perp,
\end{equation}
where $\Set{E}={\rm span}(\V{1})$, and $\Set{E}^{\perp}$ denotes the orthogonal complement of $\Set{E}$. Furthermore, when $N\geq 3$, it is known that $\Set{E}^{\perp}$ is irreducible \cite{sagan2001symmetric}. Therefore, when $N\geq 3$, the orbitals of the lines $\{\Set{L}_k\}$ would satisfy one of the following two conditions:
\begin{enumerate}[(a)]
\item All $N$ lines form a single orbital;
\item $\Set{E}$ forms a single line orbital, while all remaining $N-1$ lines form an orbital in $\Set{E}^\perp$.
\end{enumerate}

By examining the two conditions, we obtain the following central result.
\begin{theorem}[\ac{pacf} \ac{islr} controllable modulation bases]\label{thm:pacf_controllable}
For $N\geq 7$,\footnote{It is feasible to fully characterize the $N\leq 6$ case. It is deliberately neglected because the analysis is tedious and has limited practical value.} the \ac{islr} is type-based controllable, if and only if the modulation basis $\M{U}$ satisfies
\begin{equation}\label{islr_controllable_u}
\M{U}=\M{F}_N^{\rm H}\M{M}\M{C}_{\alpha,\beta},
\end{equation}
where $\M{M}\in\mathsf{Mon}(N)$, $\alpha,\beta\in[0,2\pi)$, and $\M{C}_{\alpha,\beta}$ is defined by
\begin{equation}
\M{C}_{\alpha,\beta}=\frac{e^{j\alpha}}{N}\V{1}\V{1}^{\rm T}+e^{j\beta}(\M{I}-\V{1}\V{1}^{\rm T}).
\end{equation}
\begin{IEEEproof}
Please refer to Appendix \ref{ssec:proof_pacf_controllable}.
\end{IEEEproof}
\end{theorem}

\begin{remark}
An important special case of modulation bases satisfying \eqref{islr_controllable_u} is \ac{ofdm}, corresponding to $\M{U}=\M{F}_N^{\rm H}\M{M}$ and $\M{C}_{\alpha,\alpha}=e^{j\alpha}\M{I}$. In the general \ac{islr}-controllable case, we have
\begin{equation}
{\rm ISLR}=\frac{\sum_{s\in\Set{S}}q(s)|s+(e^{j(\alpha-\beta)}-1)m_q|^4}{\big(\sum_{s\in\Set{S}}q(s)|s|^2\big)^2}-1,
\end{equation}
where $m_q=\sum_{s\in\Set{S}}q(s)s$. In light of this, if all types $\hat{q}$ used in the codebook satisfy the exact zero-mean condition $m_{\hat{q}}=0$, all \ac{islr}-controllable modulation bases would yield the same \ac{islr} values as \ac{ofdm} does.
\end{remark}

\subsection{The \ac{aacf} case}
In the \ac{aacf} case, we denote
\begin{equation}
A_{\rm a}(\V{s}) = \sum_{k=1}^{N-1} |r_k|^2,~~{\rm ISLR}=\frac{A_{\rm a}(\V{s})}{\|\V{s}\|_2^2}.
\end{equation}
If \ac{islr} is permutation invariant, for any permutation matrix $\M{\Pi}$, the conjugate action $\M{U}_{\M{\Pi}}=\M{U}\M{\Pi}\M{U}^{\rm H}$ would preserve $A_{\rm a}$, in the sense that $A_{\rm a}(\V{s})=A_{\rm a}(\M{U}_{\M{\Pi}}\V{s})$. Next, we give a characterization of this conjugate action through the following lemmata.
\begin{lemma}\label{lem:monomial_zeros}
$A_{\rm a}(\V{s})=0$ if and only if $\V{s}$ has a single non-zero entry.
\begin{IEEEproof}
Please refer to Appendix \ref{ssec:proof_monomial_zeros}.
\end{IEEEproof}
\end{lemma}

Since $\M{U}_{\M{\Pi}}$ preserves $A_{\rm a}$, it would also preserve the set of zeros characterized by Lemma \ref{lem:monomial_zeros}. Thus for each $\V{e}_k$, we have 
\begin{equation}\label{monomial_result}
\M{U}_{\M{\Pi}}\V{e}_k=\lambda_k\V{e}_{\sigma(k)}, 
\end{equation}
where $|\lambda_k|=1$, while the unitarity of $\M{U}_{\M{\Pi}}$ ensures that $\sigma(k)$ is a permutation. Therefore, Lemma \ref{lem:monomial_zeros} implies that $\M{U}_{\M{\Pi}}$ is monomial.

\begin{lemma}\label{lem:reversal}
The permutation $\sigma(k)$ is either an identity transformation or a reversal permutation.
\begin{IEEEproof}
Please refer to Appendix \ref{ssec:proof_reversal}.
\end{IEEEproof}
\end{lemma}

Lemma \ref{lem:reversal} determines the permutation component of the monomial transformation $\M{U}_{\M{\Pi}}$. Next we investigate the phase component.
\begin{lemma}\label{lem:ramp}
The coefficient $\lambda_k$ in \eqref{monomial_result} satisfies
\begin{equation}
\lambda_k = \lambda_0 q^k~{\rm or}~\lambda_{N-1}q^{N-k-1},
\end{equation}
where $|\lambda_0|=|q|=1$.
\begin{IEEEproof}
Please refer to Appendix \ref{ssec:proof_ramp}.
\end{IEEEproof}
\end{lemma}

Now, by combining Lemmata \ref{lem:monomial_zeros}, \ref{lem:reversal} and \ref{lem:ramp}, we arrive at the following result.
\begin{proposition}
When $N\geq 2$, all unitary matrices $\M{U}$ satisfying
\begin{equation}
A_{\rm a}(\M{U}\V{s})=A_{\rm a}(\V{s}),\qquad\forall\V{s}\in\mathbb{C}^N
\end{equation}
constitute the set
\begin{equation}
\Set{G}_{N} = \{a\M{D}(q)\M{R}^b:|a|=|q|=1,~b\in\{0,1\}\},
\end{equation}
where 
$$
\M{D}(q) = {\rm diag}(1,q,\dotsc,q^{N-1}),
$$
and $\M{R}$ is the reversal permutation matrix.
\end{proposition}

We may now conclude that the type-based controllability of \ac{aacf} \ac{islr} is equivalent to
\begin{equation}\label{aacf_islr_control}
\M{U}\mathscr{S}_N\M{U}^{\rm H} \subseteq \Set{G}_N.
\end{equation}
It turns out that this is impossible except for very small $N$.

\begin{theorem}[\ac{aacf} \ac{islr} is not type-based controllable]\label{thm:nogo_aislr}
When $N\geq 4$, there does not exist any unitary $\M{U}$ such that \ac{aacf} \ac{islr} is type-based controllable.
\begin{IEEEproof}
Please refer to Appendix \ref{ssec:proof_nogo_aislr}.
\end{IEEEproof}
\end{theorem}

Theorem \ref{thm:nogo_aislr} suggest that in the \ac{aacf} case, under any modulation basis, \ac{pcs} methods can only provide sidelobe control in the expectation sense, namely adjusting the \ac{eisl}, but cannot provide any codeword-level \ac{islr} guarantee. In light of this, in the rest of this paper, we focus on the \ac{pacf}.

\section{Coherent Integration via Sifted Shifting}\label{sec:ss}
In \cite{iceberg}, it is shown that conventional coherent integration over $M$ samples would yield \ac{eisl} reduction by a factor of $1/M$. In this section, we present a modified coherent integration technique for \ac{ofdm}, which provides a guaranteed $1/M$ \ac{pacf} \ac{islr} reduction when applied in conjunction with \ac{islr}-controllable type-based \ac{pcs} methods.

Broadly speaking, for each \ac{ofdm} symbol, we choose a specific circular shift $\nu$ in the frequency domain, which does not alter the \ac{islr} of each \ac{ofdm} symbol, but provides an $1/M$ reduction for the ensemble \ac{islr} over all $M$ \ac{ofdm} symbols. Note that the codeword-level \ac{islr} invariance of this technique follows directly from the fact that \ac{islr} is permutation invariant and that any circular shift is a permutation.

In particular, let us denote the \ac{pacf} sidelobe at the $k$-th delay bin of the $m$-th \ac{ofdm} symbol as $\widetilde{r}_{k,m}$, and collect them in a vector as
\begin{equation}
\V{r}_m = [r_{1,m},\dotsc,r_{N-1,m}]^{\rm T}.
\end{equation}
Accordingly, we denote
\begin{equation}
A_{{\rm p},m} = \|\V{r}_m\|_2^2,~~{\rm ISLR}_m = \frac{A_{{\rm p},m}}{\|\V{s}_m\|_2^4}.
\end{equation}
Next, denote the frequency-domain circular shift operator as
\begin{equation}
\M{Q}_\nu = {\rm diag}(e^{j2\pi \nu/N},\dotsc, e^{j2\pi (N-1)\nu/N}),
\end{equation}
the sidelobe vector after circular shift is given by
\begin{equation}
\V{r}_m^{(\nu)} = \M{Q}_{\nu} \V{r}_m.
\end{equation}
For any non-zero delay $k$, we have
\begin{equation}
\sum_{\nu=0}^{N-1}e^{j2\pi\nu/N} = \frac{1-e^{j2\pi\nu}}{1-e^{j2\pi\nu/N}} = 0,
\end{equation}
and hence
\begin{equation}
\sum_{\nu=0}^{N-1} \M{Q}_\nu = \M{0}.
\end{equation}

Now, assume that the first $m-1$ \ac{ofdm} symbols have been integrated and yield the integrated sidelobe vector $\V{q}^{(m-1)}$, while the $m$-th sidelobe vector $\V{r}_m$ is to be integrated. If we apply a frequency-domain circular shift on $\V{r}_m$, we have
\begin{equation}
\|\V{q}^{(m-1)}+\M{Q}_{\nu}\V{r}_m\|_2^2 = \|\V{q}^{(m-1)}\|_2^2+\|\V{r}_m\|_2^2+2c(\nu),
\end{equation}
with the cross term being
\begin{equation}
c(\nu) = {\rm Re}((\V{q}^{(m-1)})^{\rm H}\M{Q}_{\nu}\V{r}_m).
\end{equation}
Averaging over all shifts, we have
\begin{equation}
\frac{1}{N} \sum_{\nu=0}^{N-1} c(\nu) \!=\! {\rm Re}\left\{(\V{q}^{(m\!-\!1)})^{\rm H}\left(\frac{1}{N}\sum_{\nu=0}^{N-1}\M{Q}_{\nu}\V{r}_m\right)\right\}\!=\!0.
\end{equation}
Therefore, there exists at least one $\nu_\star$, such that
\begin{equation}
c(\nu_\star)\leq 0,
\end{equation}
otherwise the shift-averaged value cannot be zero. In this regard, we may choose an appropriate circular shift $\nu_m$ for every \ac{ofdm} symbol, such that for the complete set of $M$ \ac{ofdm} symbols we have
\begin{equation}
\|\V{q}^{(M)}\|_2^2=\left\|\sum_{m=1}^M\M{Q}_{\nu_m}\V{r}_m\right\|\leq \sum_{m=1}^M \|\V{r}_m\|_2^2.
\end{equation}
On the other hand, the mainlobe is given by $\sum_{m=1}^M\|\V{s}_m\|_2^2$. Therefore, the overall \ac{islr} is bounded by
\begin{align}
{\rm ISLR} &= \frac{\|\V{q}^{(M)}\|_2^2}{\left(\sum_{m=1}^M\|\V{s}_m\|_2^2\right)^2}\nonumber \\
&\leq \frac{\sum_{m=1}^M \|\V{s}_m\|_2^2\cdot {\rm ISLR}_m}{\left(\sum_{m=1}^M\|\V{s}_m\|_2^2\right)^2}.
\end{align}
If for each \ac{ofdm} symbol, we use an \ac{islr}-controllable \ac{pcs} to ensure ${\rm ISLR}_m\leq \gamma$,  we see that 
\begin{equation}
{\rm ISLR} \leq \frac{\gamma}{M_{\rm eff}},
\end{equation}
where $M_{\rm eff}=\frac{\left(\sum_{m=1}^M\|\V{s}_m\|_2^2\right)^2}{\sum_{m=1}^M\|\V{s}_m\|_2^2}$ is the effective number of integrated \ac{ofdm} symbols, which equals to $M$ when all the \ac{ofdm} symbols have identical energy. We call the aforementioned technique ``sifted shifting'', since all the frequency-domain circular shifts are simply sifted from $\{0,\dotsc, N-1\}$. The algorithm is summarized in Algorithm \ref{alg:ss}, where the centered power spectra are computed to filter out the mainlobe. 

\begin{remark}[Complexity and overhead]
Note that the entire $\widetilde{c}_m[\nu]$ sequence is a length-$N$ \ac{pacf}, and hence can be obtained by a single \ac{fft} and its inverse. Therefore, the overall complexity of the algorithm is on the order of $O(MN\log N)$. From the communication perspective, in order to inform the communication receiver with the circular shifts, $\lceil\log_2 N\rceil$ additional bits would have to be consumed for each \ac{ofdm} symbol.
\end{remark}

\begin{algorithm}[t]
\caption{The sifted shifting algorithm.}
\label{alg:ss}
\renewcommand{\algorithmicrequire}{\textbf{Input:}}
\renewcommand{\algorithmicensure}{\textbf{Output:}}
\begin{algorithmic}[1]
\REQUIRE \ac{ofdm} symbols $\{\V{x}_m\}_{m=1}^M$ to be transmitted, and the number of \ac{ofdm} symbols $M$.
\FOR{$m=1$ to $M$}
\STATE Compute $\V{p}_m=(\M{I}-\frac{1}{N}\V{1}\V{1}^{\rm T})|\V{x}_m|^2$;
\ENDFOR
\STATE $\nu_1=0$, $\V{p}_{\rm acc}=\V{p}_1$;
\FOR{$m=2$ to $M$}
\STATE Compute $\widetilde{c}_m[\nu]=\V{p}_{\rm acc}^{\rm H}\M{T}_\nu\V{p}_m$, where $\M{T}_\nu=\M{F}_N^{\rm H}\M{Q}_{\nu}\M{F}_N$;
\STATE Choose $\nu_m = \mathop{\arg\min}_{\nu}~\widetilde{c}_m[\nu]$;
\STATE Update $\V{p}_{\rm acc}\leftarrow \V{p}_{\rm acc}+\M{T}_{\nu_m}\V{p}_m$;
\ENDFOR
\ENSURE Shifted \ac{ofdm} symbols $\{\M{T}_{\nu_m}\V{x}_m\}_{m=1}^M$.
\end{algorithmic}
\end{algorithm}

\section{Practical \ac{isac}-Oriented \ac{pcs} Design}\label{sec:design}
\ac{pcs} was originally designed for achieving the capacity of \ac{awgn} channel under merely power constraints. Under finite constellations, the optimal input distribution is known to be the \ac{mb} distribution \cite{mb}, which is a discretized version of the zero-mean circularly symmetric complex Gaussian distribution.

One central problem associated with \ac{pcs} is that how it could be amalgamated with \acp{fec}, which is typically addressed under the framework of \ac{pas} \cite{pas}. In particular, standard \ac{pas} applies \ac{pcs} on the amplitudes of the in-phase and the quadrature (I/Q) components of \ac{qam} constellations, while leave the two sign bits out for conveying the parity bits in \ac{fec} \cite{pas}. Such a construction is valid because: 1) the parity bits in systematic \acp{fec} are approximately uniform, and 2) the \ac{mb} distribution has mutually independent I/Q components. 

However, in the context of \ac{isac}, Pareto optimal input distributions do not in general factorize to independent I/Q components. In this regard, existing works developed two classes of techniques \cite{isac_pcs_pas}. The first is to force the input distribution to factorize, which turns out to be inefficient when the sidelobe level constraints are stringent. The second, termed as \ac{gpas}, uses \ac{apsk} constellations with equal number of points on each ring, and put the parity bits on the phase of the symbols. Unfortunately, \ac{gpas} is that it is not compatible with \ac{qam}. Moreover, its equal-point ring structure, although enables phase parity bits embedding, renders the Euclidean distance distribution of the constellation less favorable, yielding potential degradation to the communication performance. 

Against this background, we would like to highlight that:
\begin{remark}[Quadrant \ac{pas} for \ac{qam}]
Parity bits can still be embedded on the signs of the I/Q components in \ac{qam}, since \ac{qam} constellations are \emph{quadrant symmetric}. To elaborate, one could apply \ac{pcs} to the constellation points in the first quadrant, while using the two uniform sign bits to indicate which actual quadrant that the symbol is mapped into. In this way, \ac{qam} are made compatible with \ac{pas}.
\end{remark}

The convenience of quadrant \ac{pas} comes with the cost that the number of constellation points in the first quadrant is considerably larger than that in the positive half of an I/Q axis. Under a given finite blocklength, it is known that distribution having larger supports will incur larger rate loss, measured by
\begin{equation}\label{rate_loss_def}
L = H(p_{\rm ref})-\frac{k}{N_{\rm b}},
\end{equation}
where $H(p_{\rm ref})$ denotes the entropy of the reference distribution (e.g., the one maximizing the \ac{mi}), $N_{\rm b}$ denotes the blocklength (in bits), and $k$ denotes the number of bits that is achievable by the \ac{pcs} distribution matcher. 

In this rest of this section, we discuss the specific designs of \ac{pcs} methods.

\subsection{\ac{ccdm}}
The \ac{ccdm} matches the reference input distribution using a single integer type, typically via arithmetic coding \cite{ccdm}. According to the discussion in Sec.~\ref{sec:islr}, ideal \ac{ccdm} can precisely control the codeword-wise \ac{islr}. Nevertheless, in practice, the types have to take integer values. If one persists in imposing hard \ac{islr} constraints, the type has to be designed through the following optimization problem
\begin{subequations}\label{klm_wc}
\begin{align}
\min_{\hat{P}\in\Set{P}_N(\Set{S})}&~~D_{\rm KL}(\hat{P}\|p)\\
{\rm s.t.}&~~\frac{\sum_{k}\hat{P}(k)|\widetilde{s}_k|^4}{\left(\sum_{k}\hat{P}(k)|\widetilde{s}_k|^2\right)^2}\leq 1+\gamma, \label{islr_type}\\
&~~p(k)=0\Rightarrow\hat{P}(k)=0,~\forall k,
\end{align}
\end{subequations}
where $\Set{S}=\{\widetilde{s}_1,\dotsc,\widetilde{s}_K\}$, $p$ denotes the reference distribution, and $D_{\rm KL}(\cdot \|\cdot)$ denotes the \ac{kld}. In general, \eqref{klm_wc} is an extremely difficult combinatorial optimization problem. A practical alternative is thus neglect the \ac{islr} constraint \eqref{islr_type}, for which the classical divergence-optimal $N$-type quantization algorithm is guaranteed to yield the optimal solution \cite{idquant}. Consequently, \ac{ccdm} would yield a single \ac{islr} value for all admissible codewords, but the resulting \ac{islr} value may not obey the design constraints. 

The natural reference distribution for \ac{ccdm} is the \ac{mi}-maximizing input distribution, the solution to the following problem
\begin{subequations}\label{pmi}
\begin{align}
\max_{p_s(s)\in\Set{P}(\Set{S})}&~~I(s;\sqrt{\rho}s+z)\\
{\rm s.t.}&~~\sum_{s\in\Set{S}}p_s(s)|s|^2\leq P_{\rm max},\label{mi_cons1}\\
&~~\frac{\sum_{s\in\Set{S}}p_s(s)|s|^4}{\left(\sum_{s\in\Set{S}}p_s(s)|s|^2\right)^2}-1\leq \kappa_{\max},\label{mi_cons2}
\end{align}
\end{subequations}
where $z\sim\mathcal{CN}(0,1)$. In general, this input distribution is quadrant-symmetric but exhibits no further symmetry, thereby yielding larger finite blocklength rate loss. In the hope of reducing the loss, one might wish to try and confining the input distribution to the exponential family, i.e.,
\begin{subequations}\label{cons_mi}
\begin{align}
\max_{\alpha,\beta}&~~I(s;\sqrt{\rho}s+z)\\
{\rm s.t.}&~~p_s(s)=\frac{1}{Z_s}\exp(-\alpha|s|^2-\beta|s|^4), \label{exp_family_form}\\
&~~\eqref{mi_cons1},~\eqref{mi_cons2},
\end{align}
\end{subequations}
which guarantees that the resulting distribution assigns equal probabilities to constellation points of the same magnitude. In what follows we shall refer to this property as ``shell degeneracy'', since all constellation points on the same constant-magnitude shell are equal-probable. Unfortunately, \ac{ccdm} by its nature cannot exploit shell degeneracy efficiently, since the sizes of each shell are in general different. Consequently, one would have to apply inconvenient variable-length coding for the phase degree of freedom. In other words, practical \ac{ccdm} cannot reduce the rate loss caused by the large size of the quadrant set.

\subsection{Joint hard-moment \ac{ess}}\label{ssec:jhmess}
In contrast to \ac{ccdm}, \ac{ess} does not explicitly match type with reference distributions, but rather confine all codewords into a set satisfying certain sample moment constraints. In the context of \ac{isac}, we consider the following constraints
\begin{subequations}\label{ess_cons}
\begin{align}
D_N&=\sum_{i=0}^{N-1} |s_i|^2 \leq P_{\max},\label{d_cons}\\
F_N&=\frac{\sum_{i=0}^{N-1} |s_i|^4}{D_N^2}-1\leq \kappa_{\max}. \label{kurtosis_cons}
\end{align}
\end{subequations}
\ac{ess} further assigns uniform probabilities to all legitimate codewords. Using the method of types argument \cite{it_cover}, when $N$ is sufficiently large, the marginal distribution of the codebook would converge to the maximal entropy distribution solving the following problem
\begin{subequations}\label{pme}
\begin{align}
\max_{p_s(s)\in\Set{P}(\Set{S})}&~~H(s)\\
{\rm s.t.}&~~\mathbb{E}\{|s|^2\}\leq P_{\max},\label{me_second}\\
&~~\mathbb{E}\{|s|^4\}\leq (1+\kappa_{\max})\mathbb{E}\{|s|^2\}.\label{me_fourth}
\end{align}
\end{subequations}
Under the same $(P_{\max},\kappa_{\max})$ constraints, the maximal entropy solution to \eqref{pme} is typically different from the one solving \eqref{pmi}, except for the large \ac{snr} limit. In order to improve the achievable \ac{mi}, we propose to first solve the constrained \ac{mi} maximization problem \eqref{cons_mi}, obtaining the optimal $(P_\star,\kappa_\star)$, and then replace the $(P_{\max},\kappa_{\max})$ in the \ac{ess} constraints \eqref{ess_cons} with $(P_\star,\kappa_\star)$. This is reasonable since the exponential family taking the form \eqref{exp_family_form} is exactly the family of entropy-maximizing distributions satisfying \eqref{me_second} and \eqref{me_fourth}. Although the problem \eqref{cons_mi} is non-convex with respect to the optimization parameters $(\alpha,\beta)$, a two-dimensional search can determine the parameters directly. 

Regarding \ac{islr}, note that all \ac{ess}-permissible codewords satisfy the hard \ac{islr} constraint \eqref{kurtosis_cons}. Therefore, \ac{ess} with an exact implementation is strictly \ac{islr} controllable. We shall refer to such exact implementations as \ac{jhmess} here after. Next, we discuss how \ac{isac}-oriented \ac{ess} could be implemented with low complexity, possibly with approximations.

\subsubsection{Exploiting shell degeneracy}
The structure of \ac{ess} makes it natural to exploit the shell degeneracy. To elaborate, in a similar spirit as that of K-ESS \cite{kess}, in \ac{jhmess} we construct a two-dimensional trellis, which is a rooted directed acyclic graph on which each path from the root to the leaves gives a set of legitimate codewords. Formally, the trellis is defined as a layered weighted graph $(\Set{V},\Set{E},\Set{W})$, where each vertex (also known as \emph{state}) is a triple such that
\begin{equation}
\Set{V} = \bigcup_{n=0}^N\{n\}\times \Set{L}_n,
\end{equation}
where $n$ denotes the layer index, while $\Set{L}_n$ denotes the set of reachable states in the $n$-th layer, recursively defined as
\begin{align}
\Set{L}_0 &= \{(0,0)\},\nonumber \\
\Set{L}_{n+1}&=\{(U+u_s,V+u_s^2):(U,V)\in\Set{L}_n,s\in\Set{S}\},
\end{align}
with the compressed state increment $u_s$ being defined as $u_s=\frac{|s|^2-2}{8}\in\mathbb{Z}_+,~\forall s\in\Set{S}$, since the energies of all \ac{qam} constellation points are multiples of $8$ plus $2$. Under this representation we have an one-to-one mapping between $(U_n,V_n)$ and $(D_n,F_n)$
\begin{equation}\label{dfuv}
D_n = 2n+8U_n,~~F_n = 4n+32U_n+64V_n.
\end{equation}
For each energy shell, we define the directed edge
$$
(n,U,V)\rightarrow (n+1,U+u_s,V+u_s^2),
$$
with the edge weight being the multiplicity of degeneracy 
$$
g_u=|\{s:s\in\Set{S},{\rm Re}(s)>0,{\rm Im}(s)>0,|s|^2-2=8u\}|.
$$
In conventional \ac{ess}, all edge weights are $1$, and each path corresponds to a single codeword. By contrast, in \ac{jhmess}, each path contributes a set of legitimate codewords whose cardinality is the product of edge weights along the path. 

Similar to conventional \ac{ess}, each state $(n,U,V)$ also records $C_n(U,V)$, the number of paths spanning from it self to all leaf nodes, which can be computed by a backward cardinality recursion from the leaf nodes as follows
\begin{equation}
C_N(U,V)\!=\!1,~C_n(U,V)\!=\!\sum_{u\in\Set{\overline{S}}}g_uC_{n+1}(U\!+\!u,V\!+\!u^2),
\end{equation}
where $\overline{S}=\{u:s\in\Set{S},|s|^2-2=8u\}$. Thus the encoding and decoding process can remain unchanged compared to the conventional \ac{ess}, except that each state would now lead to multiple equal-cost branches. 

We would like to highlight that, although the aforementioned degeneracy-aware technique greatly reduces the memory cost regarding trellis construction, it cannot reduce the finite blocklength rate loss. The reason is simple: the degeneracy-aware technique is merely a compressed representation of the same codebook. Nevertheless, \ac{ess} by its nature is known to yield less rate loss than \ac{ccdm} due to the fact that it employs all type classes satisfying the cost constraints, while \ac{ccdm} only uses one \cite{ess_short2}.

\subsubsection{Trellis construction}
In the conventional \ac{ess}, the construction of the trellis is straightforward, since one only needs to verify at each state if the current energy plus the minimal possible suffix energy is covered by the energy budget. For \ac{qam} constellations, since the minimal energy per symbol is $1$, this amounts to the simple pruning rule
\begin{equation}\label{ess_prune}
D_n+(N-n)\leq P_{\rm max}.
\end{equation}
Paths violating \eqref{ess_prune} would be pruned. However, in \ac{jhmess}, the $U$ and $V$ constraints are mutually coupled and has to be treated with extra care. In particular, for each layer-$n$ state, the suffix length is $r=N-n$, for which we define the minimal $V$-increment under a given $U$-increment $w$ as
\begin{equation}
m_r(w) = \min_{a_1,\dotsc,a_r\in\Set{S}\atop\sum_{i}u_{a_i}=w}\sum_{i=1}^r u_{a_i}^2.
\end{equation}
According to \eqref{dfuv}, given a fixed prefix $U$ and a candidate $w$, the minimal $F_N$ at layer $N$, denoted by $F$, is given by
\begin{equation}
F = F_{0,w} + 64V,
\end{equation}
where $F_{0,w}=4N + 32(U+w) + 64m_r(w)$, from which we may derive the pruning threshold:
\begin{equation}
V_n^{\max}(U)=\max_{w:U+w\leq U_{\max}} B_n(U,w),
\end{equation}
where $B_n(U,w)= \frac{\kappa_{\max}D_w^2 - NF_{0,w}}{64N}$, $D_w = 2N+8(U+w)$, and $U_{\max} = \frac{D_{\max}-2n}{8}$. Any state $(n,U,V)$ that surpass the threshold, namely $V>V_n^{\max}(U)$, will be pruned in the sense that it does not leads to further branches.

\subsubsection{Finite precision}
The finite precision implementation of the conventional \ac{ess} \cite{fpess} can be directly applied to \ac{jhmess} for a further reduction of memory cost, which will introduce a slight additional rate loss. Note that such a modification will not affect the \ac{islr} controllability, since it in effect only deletes certain legitimate codewords from the original codebook.

\subsection{Scalarized \ac{ess}}
Although the techniques presented in Sec.~\ref{ssec:jhmess} reduces the complexity of \ac{jhmess}, the memory cost can still be substantial for relatively large blocklength $N$, because the trellis spans over the large $(D_N,F_N)$-plane. In this subsection, we propose a further simplification, the \ac{sess}, reducing the memory cost significantly by merging the two constraints into a single scalar constraint, at the price of losing \ac{islr} controllability.

To elaborate, in \ac{sess}, we first determine the reference input distribution $p^\star_s(s)$ by solving the constrained \ac{mi} maximization problem \eqref{cons_mi}, and obtain the parameters $\alpha$ and $\beta$. Note that $p_s^\star(s)\propto\exp(-\alpha|s|^2-\beta|s|^4)$ is also the entropy-maximizing distribution under the moment constraints $\mathbb{E}\{|s|^2\}=P_\star$ and $\mathbb{E}\{|s|^4\}=\kappa_\star P_\star^2$. On the other hand, by defining the scalarized cost
\begin{equation}
c(s)=\alpha|s|^2+\beta|s|^4,
\end{equation}
we see that the entropy-maximizing distribution under the cost-matching constraint $\mathbb{E}\{c(s)\}\leq \mathbb{E}_{p^\star}\{c(s)\}:=c_{\max}$ is given by
\begin{equation}
p_s(s)\propto \exp(-c(s)),
\end{equation}
which is exactly $p^\star_s(s)$ itself. This motivates us to modify the conventional \ac{ess} by replacing its energy cost with the scalarized joint cost
\begin{equation}
c_{\rm S-ESS}(\V{s}) = \sum_{i=0}^{N-1}\alpha|s_i|^2+\beta|s_i|^4 = \alpha D_N+\beta F_N.
\end{equation}
In the large-$N$ limit, due to the uniqueness of the entropy-maximizing distribution, the \ac{sess} would yield identical input distribution as that of \ac{jhmess}.

However, from the perspective of \ac{islr} controllability, the \ac{sess} only guarantees
\begin{equation}
\alpha D_N+\beta F_N\leq Nc_{\max},
\end{equation}
which represents a half space in the $(D_N,F_N)$-plane, rather than the moment feasible region given by \eqref{d_cons} and \eqref{kurtosis_cons}. Therefore, \ac{sess} does not provide \ac{islr} controllability.

\begin{table}[t]
    \centering
    \caption{Comparison between \ac{isac}-oriented \ac{pcs} methods.}
    \begin{tabular}{|c|c|c|c|c|}
    \hline   \textbf{Method} & \makecell[c]{\textbf{Memory}\\\textbf{Cost}} & \makecell[c]{\textbf{Codec} \\ \textbf{Time cost}} & \makecell[c]{\textbf{Sidelobe} \\ \textbf{Controllability}} & \makecell[c]{\textbf{Rate}\\ \textbf{Loss}}\\
    \hline \ac{ccdm} & $O(N)$ & $O(N|\Set{S}|)$ & \makecell[c]{Fixed \ac{islr},\\but uncontrollable} & \Neutrey[1.3][yellow] \\
    \hline \ac{jhmess} & $O(N^3|\Set{S}|^3)$ & $O(N|\Set{S}|)$ & \ac{islr} controllable & \Smiley[1.3][green]\\
    \hline \ac{sess} & $O(N^2|\Set{S}|^2)$ & $O(N|\Set{S}|)$ & Only \ac{eisl} & \Smiley[1.3][green]\\
    \hline 1D-ESS & $O(N^2|\Set{S}|)$ & $O(N|\Set{S}|^{\frac{1}{2}})$ & Only \ac{eisl} & \Sadey[1.3][orange] \\
    \hline
    \end{tabular}
    
    \label{tab:comparison}
\end{table}

\subsection{Cost-performance analysis}\label{ssec:cpa}
Now that we have proposed three \ac{isac}-oriented \ac{pcs} schemes, in this subsection, we discuss their advantages and drawbacks. For a brief summary, please refer to Table~\ref{tab:comparison}.

\subsubsection{Memory cost}
For \ac{ccdm}, the only memory cost is the storage for each codeword, therefore is on the order of $O(N)$. For the family of \ac{ess}, the memory cost comes from the number of states in the trellis. In particular, for \ac{jhmess}, the number of states can be estimated as\footnote{We consider finite-precision implementation for all variants of \ac{ess}.}
\begin{equation}
|\Set{V}| = \sum_{n=0}^N O(nP_{\max})\times O(nP_{\max}^2)=O(N^3|\Set{S}|^3),
\end{equation}
since the number of states in each layer is on the order of the product of the energy range and the fourth-order moment range. The cubic memory cost of \ac{jhmess} can be huge for large constellations and long blocks. For \ac{sess}, the order of memory cost can be estimated as
\begin{equation}
|\Set{V}|= \sum_{n=1}^{N} \sum_{i=1}^n O(P_{\max})+ O(P_{\max}^2) = O(N^2|\Set{S}|^2).
\end{equation}
As a benchmark, we also consider the \ac{ess} implementation of the ``conventional \ac{pas}'' method in \cite{isac_pcs_pas} that force the input distribution to factorize, referred to as \emph{1D-ESS} hereafter. In 1D-ESS, the energy is accumulated along real and imaginary axes separately, thus we have
\begin{equation}
|\Set{V}|= 2\sum_{n=1}^{N} \sum_{i=1}^n O(\overline{P}_{\max})+ O(\overline{P}_{\max}^2) = O(N^2|\Set{S}|),
\end{equation}
where $\overline{P}_{\max}$ is the energy budget on the real or the imaginary axis.

\subsubsection{Encoding and decoding time complexity}
For \ac{ccdm}, the encoding and decoding time is on the order of $O(N|\Set{S}|)$ per codeword, which is standard for arithmetic coding-based \ac{ccdm} \cite{complexity_ccdm}. For all \ac{ess} variants, the encoding and decoding time is on the order of $O(NS)$, where $S$ is the number of equal-magnitude shells in the constellation. Thus for \ac{jhmess} and \ac{sess} we have $S=|\Set{S}|$, while for 1D-ESS we have $S=\sqrt{|\Set{S}|}$.

\subsubsection{Sidelobe controllability}
As discussed in previous sections, \ac{jhmess} guarantees a strict upper bound on \ac{islr}, whereas \ac{sess} provides no such guarantee. By contrast, \ac{ccdm} yields a fixed \ac{islr} because it employs a single type, however, the exact value cannot be freely specified due to its integer type approximation. Similarly, 1D-ESS cannot control \ac{islr}, as its factorized structure breaks \ac{islr} controllability.

\subsubsection{Rate loss}
Compared to \ac{ccdm}, the \ac{jhmess} and the \ac{sess} can significantly reduce the rate loss, since \ac{ess} methods can utilize multiple type classes, while \ac{ccdm} can only use one. The 1D-ESS, however, suffers from severe rate loss due to the rate panelty of forcing the input distribution to factorize, under \ac{eisl} constraints.

\section{Numerical Results}\label{sec:numerical}
In this section, we demonstrate the aforementioned analytical results and the performance of the proposed \ac{pcs} schemes using numerical simulations. If not stated otherwise, we consider 64\ac{qam} constellations, and produce $10000$ codewords (i.e., \ac{ofdm} symbols) for each configuration $4$-tuple $({\rm SNR},\kappa_{\rm off},N,{\rm AID})$, where AID denotes the algorithm identifier of the \ac{pcs} schemes (including \ac{jhmess}, \ac{sess}, \ac{ccdm}, and 1D-ESS), and $\kappa_{\rm off}$ denotes the ratio between the design \ac{eisl} and the unshaped \ac{eisl} of 64\ac{qam}, namely
\begin{equation}
\kappa_{\rm off} = \frac{\kappa_{\rm PCS}-1}{\kappa_{\rm 64QAM}-1}.
\end{equation}
We may report $\kappa_{\rm off}$ in the decibel scale, in which case it is typically negative.

Let us commence with the reference input distributions. Fig.~\ref{fig:constellations} compares the reference input distributions obtained with and without the exponential-family restriction. Recall that \ac{ccdm} can use the unconstrained \ac{mi}-maximizing reference distribution since it directly controls the type, while variants of \ac{ess} can only use the exponential-family reference distribution. Here, we consider two values of \ac{snr}, namely SNR$=10$dB and SNR$=20$dB. At both \acp{snr}, the exponential-family solutions exhibit the regular phase-independent structure induced by its cost depending only on the magnitude of symbols, whereas the unconstrained \ac{mi}-maximizing solutions clearly assign different probability values at different phase components. The distinction becomes less visible at high \ac{snr} values, where the exponential-family distributions becomes better approximations to the unconstrained \ac{mi}-maximizing distributions, since they are entropy-maximizing distributions. Later we shall see that the distinction between these two families of reference distributions only leads to a very limited difference in the achievable \ac{mi}.

\begin{figure}[t]
    \centering
    \subfloat[][Exp. family, SNR=$10$dB]{\includegraphics[width=.45\columnwidth]{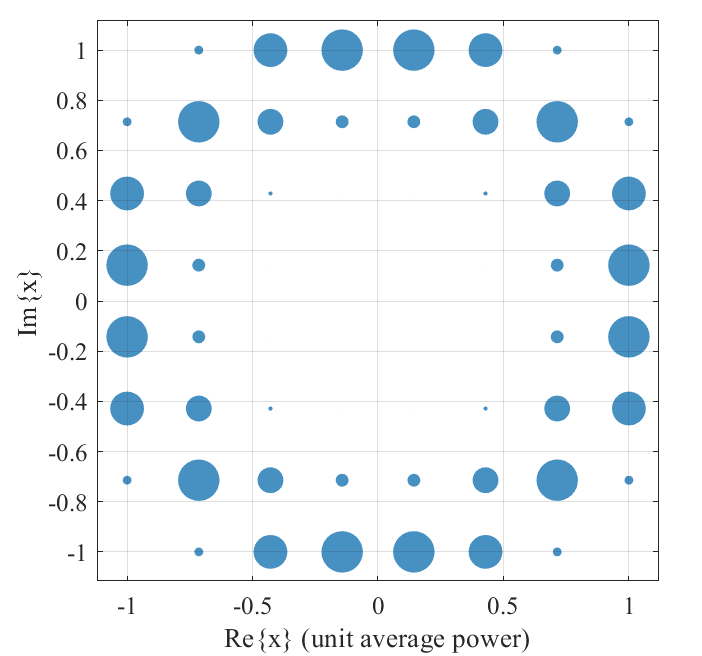}}
    \subfloat[][Unconstrained, SNR=$10$dB]{\includegraphics[width=.45\columnwidth]{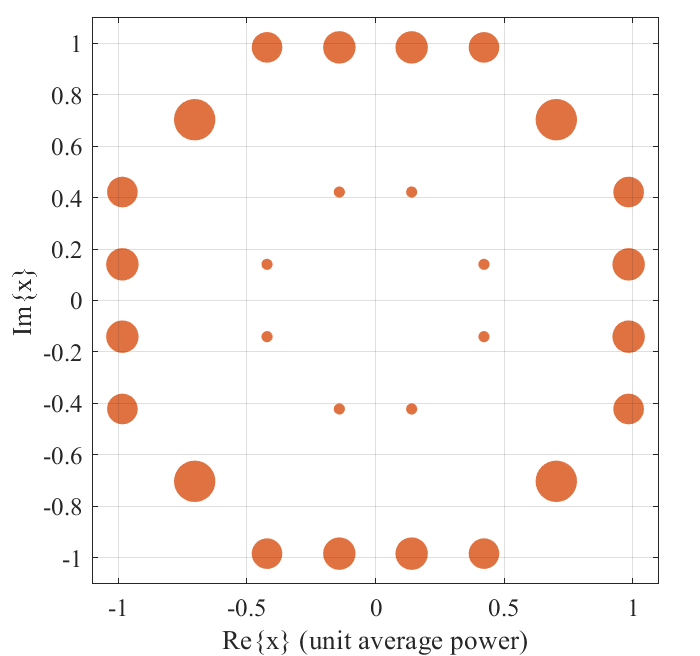}}\\
    \subfloat[][Exp. family, SNR=$20$dB]{\includegraphics[width=.45\columnwidth]{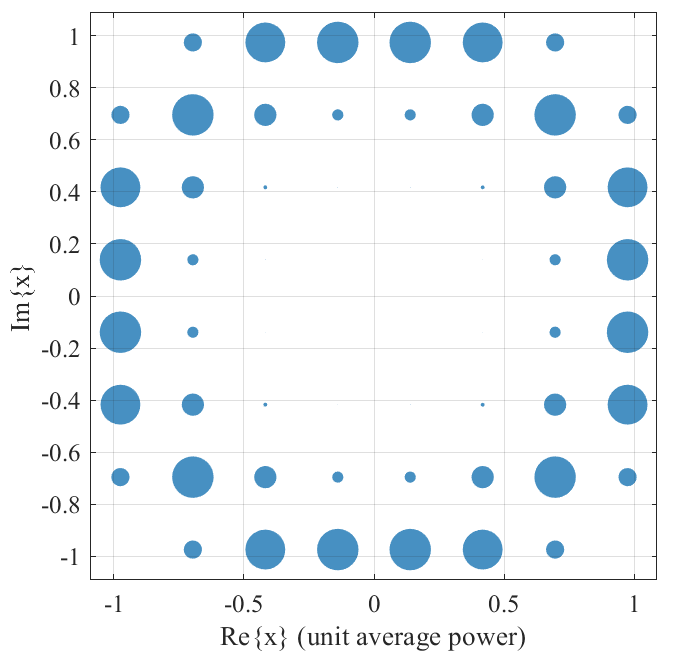}}
    \subfloat[][Unconstrained, SNR=$20$dB]{\includegraphics[width=.45\columnwidth]{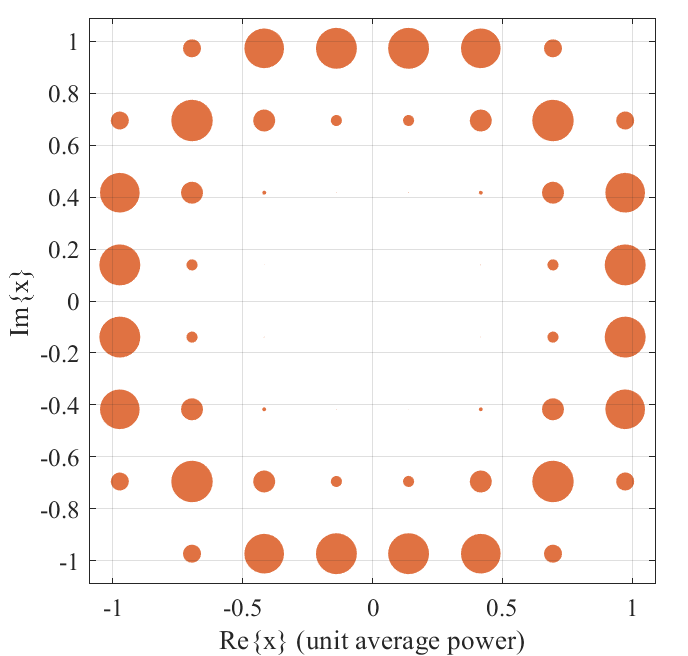}}
    \caption{Reference input distributions with and without exponential family constraints.}
    \label{fig:constellations}
\end{figure}

Next, we consider the \ac{islr} controllability of the \ac{pcs} methods. The single-codeword \acp{cdf} of \ac{islr}, portrayed in Fig.~\ref{fig:single_islr}, verify the controllability distinction discussed in Section~\ref{sec:islr}. In particular, the \ac{jhmess} \ac{cdf} terminates at the prescribed reference level, in agreement with its codeword-wise hard moment constraint. The \ac{ccdm} curve takes a single \ac{islr} value because all of its codewords belong to one type class, although the attainable level is quantized by the integer composition. By contrast, \ac{sess} and 1D-ESS retain their upper tails of \acp{cdf}, because neither construction directly constrains the fourth-order moment of every codeword. From the two clusters of curves corresponding to $\kappa_{\rm off}=-6$dB and $\kappa_{\rm off}=-10$dB, respectively, we see that the shapes of the \ac{islr} \ac{cdf} curves are almost identical across different $\kappa_{\rm off}$ values, with their mean values being shifted. The \ac{ccdm}, however, is special since its mean \ac{islr} value is not strictly controllable, as seen from its difference with the reference \ac{islr} values at different $\kappa_{\rm off}$.

We then proceed to the performance demonstration of the gain-guaranteed coherent integration by sifted shifting. To be specific, we consider the coherent integration of $M=100$ codewords, set $\kappa_{\rm off}=-10$dB, and plot the \ac{islr} \ac{cdf} curves from $1000$ different integrated \acp{pacf}. Observe that after 100 sifted-shifting integrations, Fig.~\ref{fig:ss_islr} shows a pronounced leftward shift of all \acp{cdf}, showing that coherent integration by sifted shifting indeed improves the typical sidelobe behavior. Especially, the \ac{islr} upper bound, obtained by left-shifting the single-codeword reference \ac{islr} value for $20$dB, is on the far right-hand side of all \acp{cdf}, suggesting that sifted shifting can substantially outperform its theoretical guarantee of $1/M$ \ac{islr} reduction in practice. Interestingly, the \acp{cdf} having the same single-codeword \ac{islr} mean value approximately coincide with one another after sifted shifting, while leaving the \ac{ccdm} having a significantly different \ac{islr} mean value outside the cluster. This suggests that sifted shifting tends to equate the \ac{islr} distributions, which is a phenomenon deserving further investigation.

\begin{figure}
    \centering
    \subfloat[][Single codeword]{\includegraphics[width=.45\textwidth]{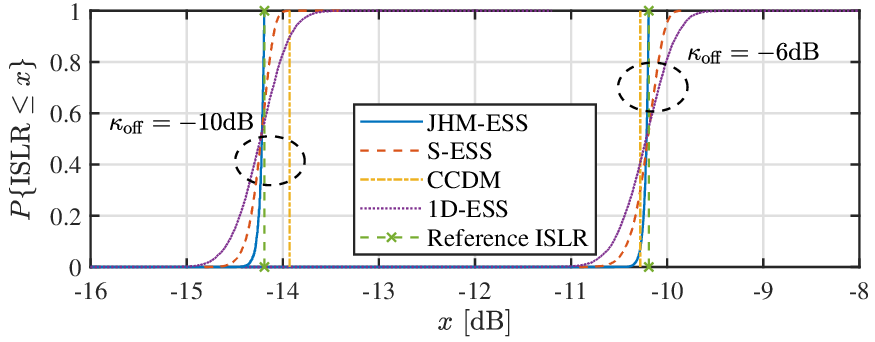}\label{fig:single_islr}}\\
    \subfloat[][$\kappa_{\rm off}=-10$dB, 100x sifted shifting coherent integration]{\includegraphics[width=.46\textwidth]{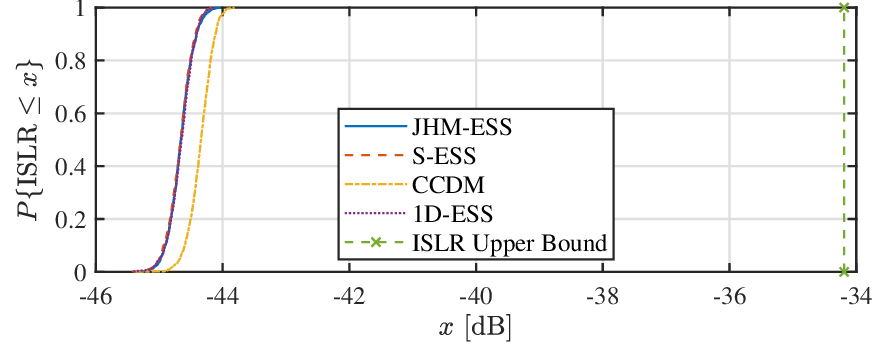}\label{fig:ss_islr}}
    \caption{\acp{cdf} of \ac{islr}, compared among \ac{pcs} schemes, with and without coherent integration.}
    \label{fig:cdf_code_uncode}
\end{figure}

Now, let us consider the finite-length rate loss of the \ac{pcs} schemes. Specifically, we consider the case where SNR=14dB, and $\kappa_{\rm off}$ values of $-4$, $-6$ and $-10$dB. The rate loss is defined as \eqref{rate_loss_def}, with $p_{\rm ref}$ being the uncontrained \ac{mi}-maximizing distribution for \ac{ccdm}, and being the exponential-family constrained distribution for all three variants of \ac{ess}. Fig.~\ref{fig:rate_loss_vs_n} shows that the finite-length rate loss clearly decreases with blocklength for all shaping schemes. Here, the blocklength is measured in bits and equals to $6N$. The two two-dimensional \ac{ess} constructions achieve the smallest losses over the simulated range, confirming the benefit of using many type classes. Furthermore, \ac{sess} is slightly better than \ac{jhmess}, because the value of $\kappa_{\rm off}$ can only determine the mean \ac{islr} value rather than the maximum \ac{islr}, and hence \ac{sess} can use even more codewords than that used by \ac{jhmess} employing hard maximal \ac{islr} constraints. \Ac{ccdm} incurs a larger composition loss because it is restricted to a single type, while 1D-ESS has the largest loss owing to the additional factorization constraint. The ordering agrees with the rate-loss discussion in Sec.~\ref{ssec:cpa} and persists as all curves improve with increasing $N$.

\begin{figure}[t]
    \centering
    \includegraphics[width=0.46\textwidth]{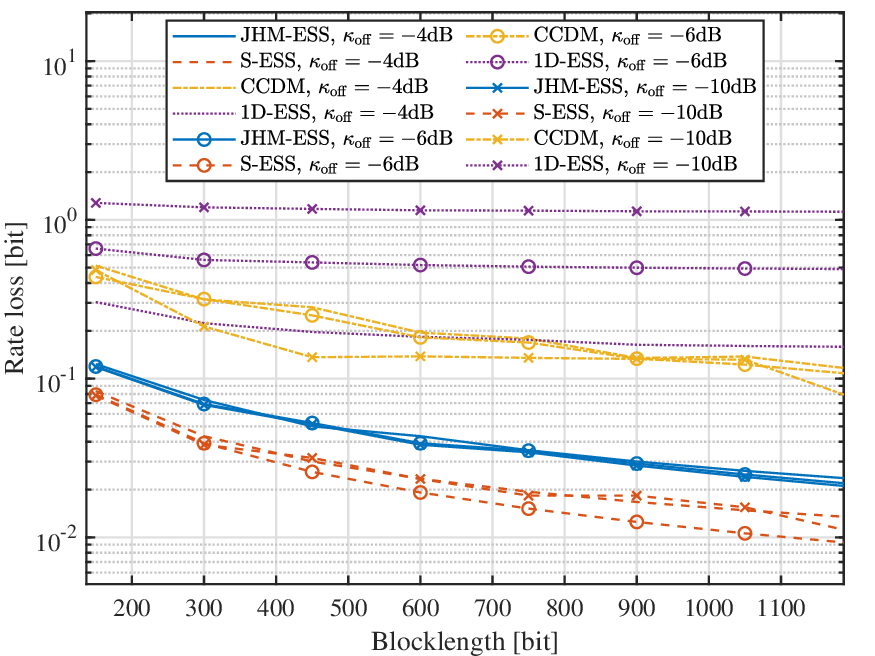}
    \caption{Rate loss vs. blocklength (in bits), compared to the unconstrained \ac{mi}-maximizing reference input distribution.}
    \label{fig:rate_loss_vs_n}
\end{figure}

For the demonstration of the rate-\ac{islr} tradeoff, we use the parameters of $N=300$ and SNR=14dB, and adjust the values of $\kappa_{\rm off}$ to yield different values of expected \ac{islr} and maximum \ac{islr}. Especially, the ``expectation'' and ``maximum'' are computed out of the $10000$ codewords for each \ac{pcs} scheme. When the horizontal axis is the mean \ac{islr}, from Fig.~\ref{fig:rate_mean_islr} we see that the two-dimensional \ac{ess} schemes exhibit nearly identical uncoded (i.e., without \ac{fec}) rate trends, with a clear gap from \ac{ccdm}, whereas 1D-ESS is substantially less efficient than all the remaining schemes. The situation is a bit different when we consider the achievable empirical \ac{mi} (computed from the $10000$ codewords for each \ac{pcs} scheme), for which we see from the zoomed-in inset that the \ac{ccdm} is the supremum among all candidates. This is because the empirical \ac{mi} does not account for the finite-length rate loss, and \ac{ccdm} enjoys its unconstrained \ac{mi} maximization. When the horizontal axis is the max \ac{islr}, we see that \ac{sess} is slightly less favorable than \ac{jhmess} due to its lack of control of the maximum \ac{islr}, while the 1D-ESS becomes even worse, which corroborates the corresponding \ac{cdf} curves in Fig.~\ref{fig:single_islr}.

\begin{figure}
    \centering
    \subfloat[][Rate vs. mean \ac{islr}]{\includegraphics[width=.46\textwidth]{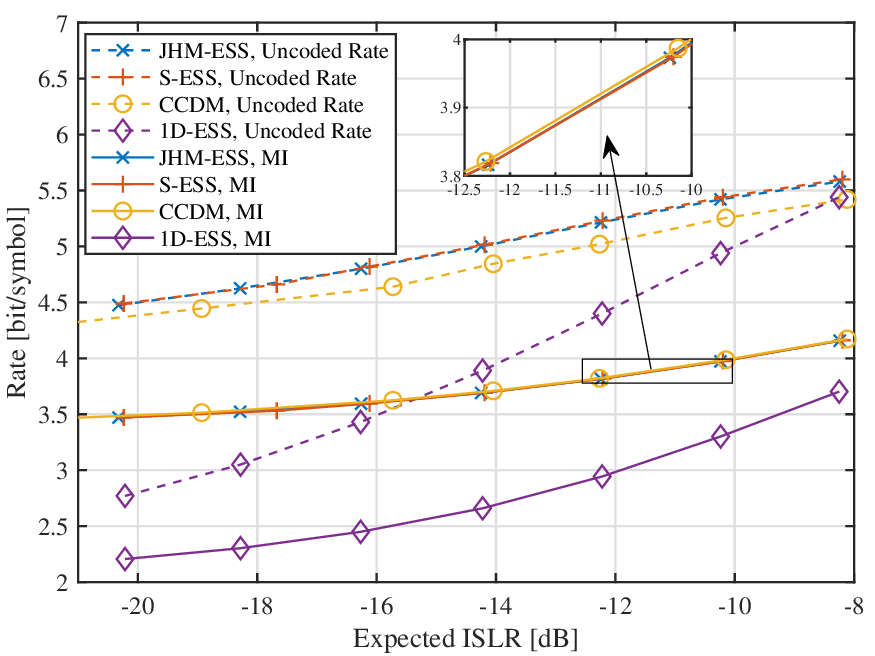}\label{fig:rate_mean_islr}}\\
    \subfloat[][Rate vs. max \ac{islr}]{\includegraphics[width=.46\textwidth]{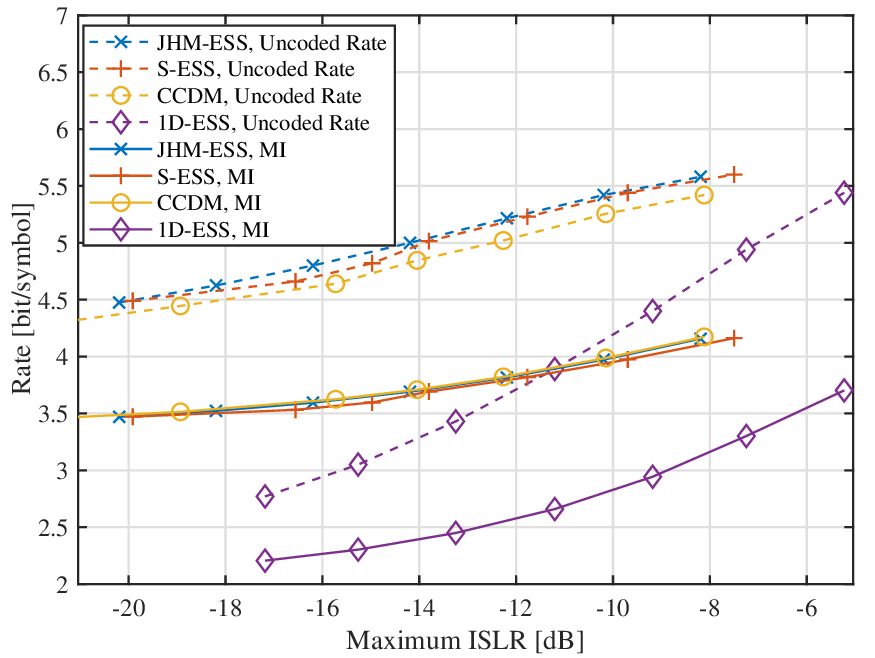}\label{fig:rate_max_islr}}
    \caption{The tradeoff between achievable rate and \ac{islr}.}
    \label{fig:rate_islr}
\end{figure}

Finally, let us consider the combination of \ac{pcs} schemes and \acp{fec}, via the \ac{pas} framework. In particular, we consider the \ac{ldpc} codes having blocklengths of 48, 1296, 1944 bits, and code rates of $1/2$, $2/3$, $3/4$, $5/6$, prescribed by the IEEE 802.11 standard \cite{IEEE80211_2024}. Due to the design constraint of \ac{pas}, the rate $1/2$ \ac{fec} is in fact invalid for 64\ac{qam}, but will later be employed in other benchmarks. 

We first investigate the \ac{bler} of the coded \ac{pcs} schemes, as portrayed in Fig.~\ref{fig:bler_k10}, where we consider the specific configuration of blocklength 648, $\kappa_{\rm off}=-10$dB, and code rate $2/3$. We observe that \ac{jhmess} and \ac{sess} operate at closely comparable $E_b/N_0$, while \ac{ccdm} incurs only a modest displacement. The 1D-ESS curve is shifted farther to the right, reflecting the rate and distribution mismatch introduced by separate shaping of the two real dimensions. These results indicate that enforcing the joint hard-moment constraint does not introduce a substantial coding penalty relative to the strongest finite-length baselines, and are in alignment with the aforementioned uncoded finite-length rate loss results.

\begin{figure}
    \centering
    \includegraphics[width=0.45\textwidth]{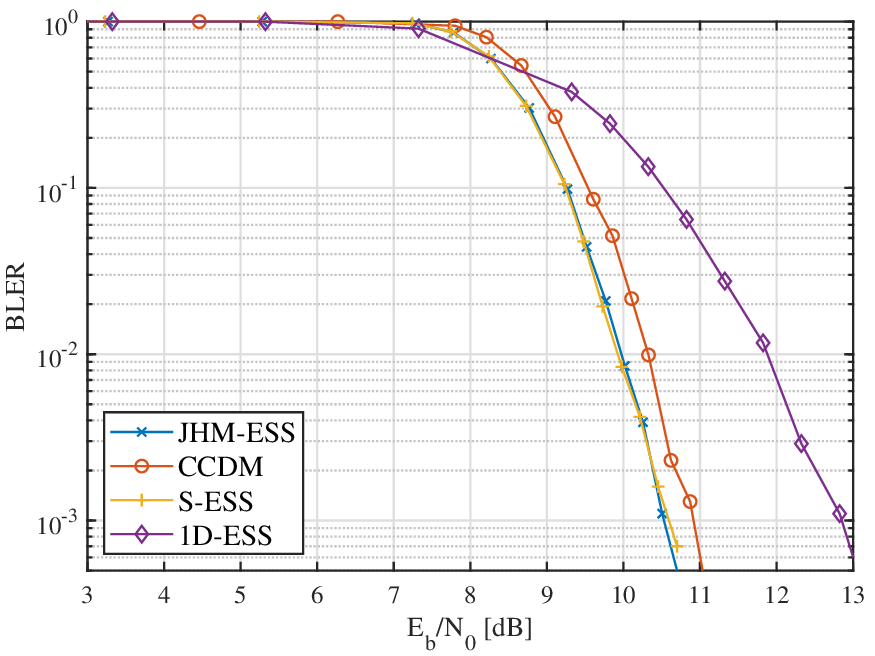}
    \caption{\Ac{bler} of coded \ac{pcs} schemes.}
    \label{fig:bler_k10}
\end{figure}

We next investigate the largest rate supported under varying \ac{fec} code rate and fixed \ac{bler} upper bound of $0.01$, with $\kappa_{\rm off}=-10$dB. In this case, we consider two types of theoretical upper bounds of the achievable rate, namely the \ac{mi} and the \ac{fbl} achievable rate (under normal approximation) given by \cite{fbl}
\begin{equation}
R_{\rm FBL}=I(\V{s};\V{y})-\sqrt{\frac{V}{N}}Q^{-1}(\epsilon),
\end{equation}
where $\epsilon=0.01$ is the \ac{bler} requirement, $Q^{-1}(\cdot)$ denotes the inverse of the standard Gaussian-$Q$ function, and $V$ denotes the channel dispersion that can be computed from the reference distribution $p_{\rm ref}(s)$ via
$$
V = {\rm Var}_{p_{\rm ref}}\left\{\log\frac{p_{y|s}(y|s)}{p_y(y)}\right\}.
$$
The \ac{fbl} bound can be much tighter than \ac{mi} in the \ac{fbl} regime. For additional benchmarks, we consider the 64PSK having the absolute zero \ac{islr} as a lower benchmark, and the 64APSK \ac{gpas} proposed in \cite{isac_pcs_pas}. We use the conventional \ac{ess} as the amplitude shaper in the 64APSK \ac{gpas} to ensure minimal rate loss, and hence fairness of comparison.

In Fig.~\ref{fig:rate_and_bound_k10} we plot the achievable rate at BLER=0.01 for all candidate schemes, with the numbers beside the markers representing the corresponding \ac{fec} code rate. We see that the \ac{sess} exhibits the smallest gap from the \ac{fbl} rate bound, with a clear gap from the \ac{ccdm}. The 64APSK \ac{gpas} exhibits a comparable performance as that of \ac{sess} at high \ac{snr} values, whereas its low-\ac{snr} performance is significantly lower, which corroborates the fact that the Euclidean distance distribution of the 64APSK constellation is less favorable. Notably, the 64PSK can even outperform the \ac{mi} bound of factorized (``1D'') input distributions, suggesting that 1D-ESS is not practically favorable at $\kappa_{\rm off}=-10$dB. This also aligns with the observation in \cite{isac_pcs_pas} that factorized input distributions yield large rate loss in the small-$\kappa$ regime. 

\begin{figure}
    \centering
    \includegraphics[width=0.45\textwidth]{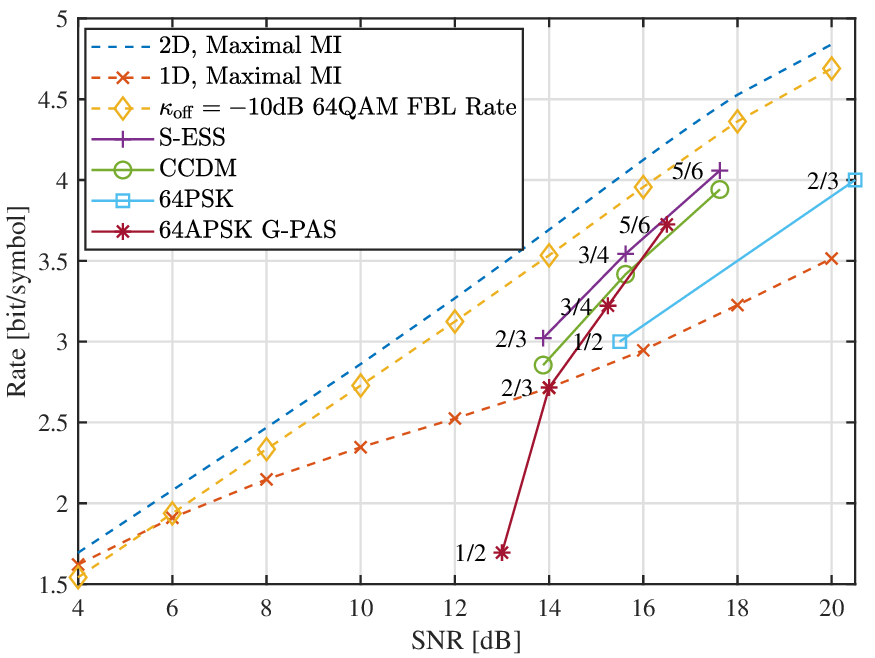}
    \caption{Achievable rate of various \ac{pcs} schemes under \ac{bler}$\leq 0.01$, and the corresponding bounds.}
    \label{fig:rate_and_bound_k10}
\end{figure}

\section{Conclusions}\label{sec:conclusions}
This paper investigated when and how type-based \ac{pcs} can control ranging sidelobes at the codeword level, rather than only in expectation. By identifying type-based controllability with permutation invariance, we characterized all unitary modulation bases that permit \ac{pacf} \ac{islr} control for $N\geq 7$. The resulting family contains \ac{ofdm} as a principal special case and yields the same \ac{islr} as \ac{ofdm} for exactly zero-mean types. By contrast, we proved that no unitary modulation basis can provide type-based control of \ac{aacf}-based \ac{islr} for $N\geq 4$. We also developed sifted shifting for coherent integration, which preserves the individual type-based guarantees and ensures an integrated \ac{islr} upper bounded by $\gamma/M_{\rm eff}$.

To translate these structural results into practical signaling schemes, we introduced quadrant \ac{pas} for joint shaping over \ac{qam} and developed two \ac{ess} constructions. In particular, \ac{jhmess} uses a degeneracy-aware two-dimensional trellis to impose hard moment constraints, thereby guaranteeing the maximum \ac{islr}, whereas \ac{sess} scalarizes the two moments to reduce the memory order from $O(N^3|\Set{S}|^3)$ to $O(N^2|\Set{S}|^2)$, but controls only ensemble behavior. The numerical results corroborated this tradeoff. The two \ac{ess} schemes achieved smaller finite-length rate loss than \ac{ccdm} and factorized 1D-ESS under the same \ac{islr} constraint.

To summarize, our findings establish a direct link between modulation symmetry, type-based shaping, and codeword-wise sensing guarantees, and show that deterministic sidelobe control can be obtained without hurting the implementational convenience of constellation shaping.

\appendices
\section{Proofs of Results in Section~\ref{sec:islr}}\label{sec:proof_islr}
\subsection{Proof of Proposition \ref{prop:perm}}\label{ssec:proof_perm}
\begin{IEEEproof}
First note that types themselves are permutation invariant, thereby \eqref{controllability} would imply 2). Conversely, if 2) holds, for any realizable type $q$, we may choose arbitrarily a representative codeword $\V{s}$ satisfying $\hat{P}_{\V{s}}=q$, and define $h_{\M{U}}(q)=J_{\M{U}}(\V{s})$. If we choose another representative codeword $\V{t}$ in the type class $\Set{T}_N(q)$, there exists a permutation matrix $\M{\Pi}$ such that $\V{t}=\M{\Pi}\V{s}$ since all sequences in the same type class differ only by their order of entries. Thus 2) implies $h_{\M{U}}(\V{t})=h_{\M{U}}(\V{s})$ for all $\V{t}$ and $\V{s}$ in the same type class, which in turn implies 1).
\end{IEEEproof}

\subsection{Proof of Lemma \ref{lem:monomial}}\label{ssec:proof_monomial}
\begin{IEEEproof}
Let $\V{u}_n$ denote the $n$-th column of $\M{U}$. The unitarity and $\ell_4$ isometry imply
\begin{equation}\label{unitary_isometry}
\|\V{u}_n\|_2^2=1,~\|\V{u}\|_4^4=1.
\end{equation}
In addition, we have
\begin{equation}\label{fourth_expansion}
\|\V{u}_n\|_2^4=\|\V{u}_n\|_4^4+2\sum_{i<j}|u_{in}|^2|u_{jn}|^2.
\end{equation}
Combining \eqref{unitary_isometry} and \eqref{fourth_expansion} yields $\sum_{i<j}|u_{in}|^2|u_{jn}|^2=0$, which implies that each column has at most one non-zero entry. Using again the unitarity, we see that $\M{U}$ has to be monomial, and we are done.
\end{IEEEproof}

\subsection{Proof of Theorem \ref{thm:pacf_controllable}}\label{ssec:proof_pacf_controllable}
\begin{IEEEproof}
We examine the two conditions regarding the orbitals of lines separately.
\subsubsection{Condition (a)}
Without loss of generality, assume that $\Set{L}_0,\dotsc,\Set{L}_{N-1}$ constitute an orbital. On each line $\Set{L}_k$, we choose a unit vector $\V{v}_k$. Since the lines are on the same orbital, the symmetric group acts transitively on these lines, namely,
\begin{equation}
\exists \M{G}\in\mathscr{S}_N,~\M{G}\Set{L}_k=\Set{L}_j,
\end{equation}
for any $j$, $k$. Therefore, for any $\M{G}\in\mathscr{S}_N$, there exists a complex scalar $|\lambda_{\M{G},k}|=1$ such that $\M{G}\V{v}_k=\lambda_{\M{G},k}\V{v}_{\sigma_{\M{G}}(k)}$. On the other hand, being a permutation matrix, any $\M{G}\in\mathscr{S}_N$ satisfies $\M{G}\V{1}=\V{1}$. The unit vector $\V{u}=\frac{1}{\sqrt{N}}\V{1}$ thus satisfies
\begin{equation}
\V{u}^{\rm H}\V{v}_k=(\M{G}\V{u})^{\rm H}\V{v}_k = \lambda_{\M{G}_k}\V{u}^{\rm H}\V{v}_{\sigma_{\M{G}}(k)}.
\end{equation}
Taking the magnitude from both sides yields
\begin{equation}
|\V{u}^{\rm H}\V{v}_k| = |\V{u}^{\rm H}\V{v}_{\sigma_{\M{G}}(k)}|,
\end{equation}
thereby $\V{u}$ has constant-modulus projections on all the unit vectors $\{\V{v}_k\}$. By adjusting the phases of $\V{v}_k$'s we obtain
\begin{equation}
\V{u}=\frac{1}{\sqrt{N}}\sum_{k=0}^{N-1}\V{v}_k.
\end{equation}
Since $\M{G}\V{u}=\V{u}$, it further holds that $\lambda_{\M{G},k}=1$. Therefore the actions of $\mathscr{S}_N$ elements on the lines constitute a natural permutation representation $\tau:\mathscr{S}_N\rightarrow \mathscr{S}_N$, which is an automorphism. For $N\neq 6$, it is known that all automorphisms of $\mathscr{S}_N$ are inner automorphisms \cite[Sec.~8.2]{permgroup}, which implies that there exists a permutation matrix $\M{\Xi}$ such that
\begin{equation}\label{inner_auto}
\tau(\M{G})=\M{\Xi}\M{G}\M{\Xi}^{\rm T},\qquad \forall\M{G}\in\mathscr{S}_N.
\end{equation}
Now, let $\M{B}=[\V{v}_0,\dotsc,\V{v}_{N-1}]$, we see that there exists a unitary diagonal matrix $\M{D}$ such that $\M{B}=\M{V}^{\rm H}\M{D}$, since $\V{v}_k\in\M{V}^{\rm H}{\rm span}(\V{e}_k)$. Choosing an appropriate $\sigma$ such that $\tau(\M{G})\V{e}_k=\V{e}_{\sigma_{\M{G}}(k)}$, we have
\begin{align}
\M{G}\M{B}\V{e}_k&=\M{G}\M{V}^{\rm H}\M{D}\V{e}_k\nonumber\\
&=[\V{v}_{\sigma_{\M{G}}(0)},~\dotsc,~\V{v}_{\sigma_{\M{G}}(N-1)}]\M{D}\V{e}_k\nonumber \\
&=\M{V}^{\rm H}\M{D}\tau(\M{G})\V{e}_k\nonumber \\
&=\M{B}\tau(\M{G})\V{e}_k,
\end{align}
and hence $\M{G}\M{B}=\M{B}\tau(\M{G})$. Using \eqref{inner_auto}, we obtain
\begin{equation}
\M{G}\M{B}\M{\Xi}=\M{B}\M{\Xi}\M{G}.
\end{equation}
This implies that $\M{B}\M{\Xi}$ is a centralizer of $\M{G}\in\mathscr{S}_N$. The following lemma gives a characterization of such centralizers.
\begin{lemma}\label{lem:commutation}
If a unitary matrix $\M{C}$ commutes with all permutation matrices in $\mathscr{S}_N$, we have
\begin{equation}\label{c_ab}
\M{C}=\M{C}_{\alpha,\beta}:=e^{j\beta}\M{I}+\frac{1}{N}(e^{j\alpha}-e^{j\beta})\V{1}\V{1}^{\rm T}.
\end{equation}
\begin{IEEEproof}
Using the classical orbital basis result on the centralizer algebra of permutation representations \cite[Sec.~2.2]{cameron_permgroup}, we obtain
\begin{equation}
\M{C}\in{\rm span}\{\M{I},\V{1}\V{1}^{\rm T}\}.
\end{equation}
Using the unitarity of $\M{C}$, it follows that the eigenvalues on ${\rm span}(\V{1})$ and its orthogonal complement are both unimodular, thus leading to \eqref{c_ab}.
\end{IEEEproof}
Using Lemma \ref{lem:commutation}, we finally obtain
\begin{equation}
\M{V}=\M{D}\M{\Xi}\M{C}_{\alpha,\beta}^{\rm H}=\M{D}\M{\Xi}\M{C}_{-\alpha,-\beta},
\end{equation}
where $\M{D}\M{\Xi}\in\mathsf{Mon}(N)$. 
\end{lemma}

\subsubsection{Condition (b)}
In this case, the fact that the orbital contains $N-1$ lines implies that there is a homomorphism $\varphi:\mathscr{S}_N\rightarrow\mathscr{S}_{N-1}$. When $N\geq 5$, $\mathscr{S}_N$ has a non-trivial normal subgroup, the alternating group $\mathscr{A}_N$. Since $\varphi$ is not injective, we have $\mathscr{A}_N\subseteq {\rm Ker}~\varphi$ for $N\geq 5$. Thus using the basic theorem of homomorphisms, we have
\begin{equation}
|{\rm Im}~\varphi|\leq \frac{|\mathscr{S}_N|}{|\mathscr{A}_N|}=2.
\end{equation}
However, in Condition (b) the orbital size argument yields $|{\rm Im}~\varphi|\geq N-1$, which is clearly impossible for $N\geq 5$. Therefore, orbitals satisfying Condition (b) do not exist when $N\geq 5$.

The full proof is thus completed by combining the analysis of Conditions (a) and (b).
\end{IEEEproof}

\subsection{Proof of Lemma \ref{lem:monomial_zeros}}\label{ssec:proof_monomial_zeros}
\begin{IEEEproof}
The sufficiency is straightforward since the sidelobes of a sequence having a single non-zero entry are obviously all zero. Conversely, assume by contradiction that $\V{s}$ has at least two non-zero entries, let $i<j$ be the smallest and the largest non-zero indices. We see that for $k=j-i$, the autocorrelation $r_k=\sum_{n} s_n^*s_{n+k}$ is non-zero since the product $s_n^*s_{n+k}$ is only non-zero when $n=i$. This implies that $A_{\rm a}(\V{s})>0$, and hence the original assumption is false. This completes the proof.
\end{IEEEproof}

\subsection{Proof of Lemma \ref{lem:reversal}}\label{ssec:proof_reversal}
\begin{IEEEproof}
Note that
\begin{align}
\|\V{s}\|_4^4 &= \|\V{s}\|_2^2+A_{\rm a}(\V{s})\nonumber \\
&=\sum_{a+b=c+d\atop 0\leq a,b,c,d<N}s_as_bs_c^*s_d^*. \label{fourth_monomials}
\end{align}
Preservation of $A_{\rm a}$ is thus equivalent to the preservation of $\|\V{s}\|_4$ since $\|\V{s}\|_2^4$ is already preserved due to unitarity. Now, substituting \eqref{monomial_result} into \eqref{fourth_monomials}, we obtain
\begin{equation}
\sigma(a)+\sigma(b)=\sigma(c)+\sigma(d).
\end{equation}
Using $0+n=1+(n-1)$, we further obtain
\begin{equation}
\sigma(n)-\sigma(n-1)=\sigma(1)-\sigma(0):=a,
\end{equation}
which further implies that
\begin{equation}\label{affine_perm}
\sigma(n)=\sigma(0)+a.
\end{equation}
Since $\sigma(k)$ is a bijection, $a$ is a non-zero integer. Thus \eqref{affine_perm} implies $(N-1)|a|\le N-1$, thereby $a=\pm 1$. Thus we have
\begin{equation}
\sigma(n)=n~{\rm or}~\sigma(n)=N-1-n,
\end{equation}
which corresponds to the identity and the reversal permutation, respectively. Thus the proof is completed.
\end{IEEEproof}

\subsection{Proof of Lemma \ref{lem:ramp}}\label{ssec:proof_ramp}
\begin{IEEEproof}
First assume that the permutation component is an identity transformation. Using again the identity $0+n=1+(n-1)$ we have the monomial $s_0s_ns_1^*s_{n-1}^*$. Note that all four permitted permutations $(0,n,1,n-1),~(n,0,1,n-1),~(0,n,n-1,1),~(n,0,n-1,1)$ correspond to the same monomial $s_0s_ns_1^*s_{n-1}^*$, and hence we have $4\lambda_0\lambda_n\lambda_1^*\lambda_{n-1}^*=4$. This implies
\begin{equation}
\frac{\lambda_n}{\lambda_{n-1}} = \frac{\lambda_1}{\lambda_0} := q.
\end{equation}
When the permutation component is the reversal permutation, we may simply alter the indices accordingly. The proof is thus completed.
\end{IEEEproof}

\subsection{Proof of Theorem \ref{thm:nogo_aislr}}\label{ssec:proof_nogo_aislr}
\begin{IEEEproof}
Let
\begin{equation}
\widetilde{\Set{G}}_N=\{a\M{D}(q):|a|=|q|=1\},
\end{equation}
which is an Abelian normal subgroup of $\Set{G}_N$. Note that
\begin{equation}
\Set{G}_N/\widetilde{\Set{G}}_N \cong \mathscr{C}_2,
\end{equation}
where $\mathscr{C}_2$ is the 2-cyclic group accounting for the potential reversal. If \eqref{aacf_islr_control} holds, the commutator subgroup of $\mathscr{S}_N$, the alternating group $\mathscr{A}_N$, would have to be mapped into $\widetilde{\Set{G}}_N$ by the conjugate mapping $g:\M{G}\mapsto \M{U}\M{G}\M{U}^{\rm H}$. However, when $N\geq 4$, $\mathscr{A}_N$ is non-Abelian and hence cannot be mapped into any Abelian group via the injection $g$. Thus we are done.
\end{IEEEproof}

\bibliographystyle{IEEEtran}
\bibliography{isacpcs}
\end{document}